\documentclass[11pt]{amsart}
\usepackage{amssymb,amsmath,fullpage,times}
\usepackage{mathrsfs}
\usepackage{amssymb}
\usepackage{graphics}
\usepackage{graphicx}
\usepackage{subfigure}
\usepackage{float}
\usepackage{bm}
\usepackage{amsfonts}
\usepackage{color}

\allowdisplaybreaks

\newtheorem{theorem}{Theorem}

\newtheorem{lemma}[theorem]{Lemma}

\newcommand{\mi}{\mathrm{i}}

\def\lam{\lambda}

\def\Re{\mathrm{Re}}
\def\lim{\mathrm{lim}}

\def\det{\mathrm{det}}
\def\arg{\mathrm{arg}}

\def\eref#1{(\ref{#1})}

\usepackage[square,numbers,sort&compress]{natbib}
\def\Im{\mathrm{Im}}

\begin{document}

\date{}

\title{Darboux transformation and soliton solutions \\ of the semi-discrete massive Thirring model}
\author{Tao Xu}
\address[Tao Xu]{State Key Laboratory of Heavy Oil Processing, China University of Petroleum, Beijing 102249, China and
College of Science, China University of Petroleum, Beijing 102249, China}
\author{Dmitry E. Pelinovsky}
\address[Dmitry E. Pelinovsky]{Department of Mathematics and Statistics, McMaster University, Hamilton, Ontario, Canada, L8S 4K1 and
Department of Applied Mathematics, Nizhny Novgorod State Technical University, 24 Minin street, 603950 Nizhny Novgorod, Russia}

\maketitle

\begin{abstract}
A one-fold Darboux transformation between solutions of the semi-discrete massive Thirring model is derived
using the Lax pair and dressing methods. This transformation is used to find the exact expressions for 
soliton solutions on zero and nonzero backgrounds. It is shown that the discrete solitons have the same 
properties as solitons of the continuous massive Thirring model.
\end{abstract}

%\noindent{{\bf Keywords:} Massive Thirring model, Integrable semi-discretization, Darboux transformation, Solitons}

% Possible journals: (1) Phys. Lett. A, (2) J. Nonlin. Math. Phys., (3) J Phys. A

\section{Introduction}

The massive Thirring model (MTM) in laboratory coordinates is an example of the nonlinear Dirac equation
arising in two-dimensional quantum field theory \cite{Thirring}, optical Bragg gratings \cite{Sipe}, and
diatomic chains with periodic couplings \cite{Alexeeva}. This model received much of attention
because of its integrability \cite{Mikhailov} which was used to study the inverse scattering \cite{KN,KMI,KW,KM,PelSaal,Villarroel,W},
soliton solutions \cite{O,BG87,BG91,BG93}, spectral and orbital stability of solitons \cite{KL,Yusuke1,Yusuke2},
and construction of rogue waves \cite{Deg}.

Several integrable semi-discretizations of the MTM in characteristic coordinates were proposed in the literature \cite{Nijhoff1,Nijhoff2,Tsuchida,Tsuchida2010,Tsuchida-preprint} by discretizing one of the two characteristic coordinates.
These semi-discretizations are not relevant for the time-evolution problem related to the MTM in laboratory coordinates.
It was only recently \cite{JoshiJPA} when the integrable semi-discretization of the MTM in laboratory coordinates was derived.
The corresponding semi-discrete MTM is written as the following system of three coupled equations:
\begin{equation}
\label{MTM-discrete}
\left\{ \begin{array}{l}
\displaystyle
4 \mi \frac{d U_n}{dt} + Q_{n+1} + Q_n  + \frac{2 \mi}{h} (R_{n+1}-R_n) + U_n^2 (\bar{R}_n + \bar{R}_{n+1}) \\
\qquad \qquad - U_n (|Q_{n+1}|^2 + |Q_n|^2 + |R_{n+1}|^2 + |R_n|^2) - \frac{\mi h}{2} U_n^2 (\bar{Q}_{n+1}-\bar{Q}_n)  = 0, \\
\displaystyle
-\frac{2 \mi}{h} (Q_{n+1}-Q_n) + 2 U_n - |U_n|^2 (Q_{n+1} + Q_n) = 0, \\
\displaystyle
R_{n+1} + R_n - 2 U_n + \frac{\mi h}{2} |U_n|^2 (R_{n+1} - R_n) = 0, \end{array} \right.
\end{equation}
where $h$ is the lattice spacing of the spatial discretization and $n$ is the discrete lattice variable. 
$\bar{R}_n$ and $\bar{Q}_n$ denote the complex conjugate of $R_n$ and $Q_n$ respectively.
Only the first equation of the system (\ref{MTM-discrete}) represents the time evolution problem, whereas
the other two equations represent the constraints which define components of
$\{ R_n \}_{n \in \mathbb{Z}}$ and $\{ Q_n \}_{n \in \mathbb{Z}}$
in terms of $\{ U_n \}_{n \in \mathbb{Z}}$ instantaneously in time $t$.

In the continuum limit $h \to 0$, the slowly varying solutions to the system (\ref{MTM-discrete}) can be represented by
$$
U_n(t) = U(x=hn,t), \quad R_n(t) = R(x = hn,t), \quad Q_n(t) = Q(x=nh,t),
$$
where the continuous variables satisfy the following three equations:
\begin{equation}
\label{MTM-continuous}
\left\{ \begin{array}{l}
\displaystyle
2 \mi \frac{\partial U}{\partial t} + \mi \frac{\partial R}{\partial x} + Q + U^2 \bar{R} - U (|Q|^2 + |R|^2) = 0, \\ [2mm]
\displaystyle
-\mi \frac{\partial Q}{\partial x} + U - |U|^2 Q = 0, \\ [2mm]
\displaystyle
R - U = 0. \end{array} \right.
\end{equation}
The system (\ref{MTM-continuous}) in variables $U(x,t) = u(x,t-x)$ and $Q(x,t) = v(x,t-x)$ yields
the continuous MTM system in the form:
\begin{equation}
\left\{
\begin{aligned}
\mi \left( \frac{\partial u}{\partial t}+\frac{\partial u}{\partial x} \right) + v =|v|^2u,\\
\mi \left( \frac{\partial v}{\partial t}-\frac{\partial v}{\partial x} \right) + u =|u|^2v.\\
\end{aligned}
\right.
\label{MTM}
\end{equation}

It is shown in \cite{JoshiJPA} that the semi-discrete MTM system (\ref{MTM-discrete}) is the
compatibility condition
\begin{equation}
\label{Lax-discrete}
\frac{d}{dt} N_n(\lambda) = P_{n+1}(\lambda) N_n(\lambda) - N_n(\lambda) P_n(\lambda),
\end{equation}
of the following Lax pair of two linear equations:
\begin{subequations}
\begin{align}
& \Phi_{n+1}(\lambda) = N_n(\lambda) \Phi_n(\lambda),  \quad
N_n(\lambda) =
\begin{pmatrix}
 \lambda +\frac{2 \mi}{h \lambda} \left( \frac{1+\frac{\mi}{2} h |U_n|^2}{1-\frac{\mi}{2} h |U_n|^2} \right) &
 \frac{2 U_n}{1-\frac{\mi}{2} h |U_n|^2} \\[2mm] \frac{2 \bar{U}_n}{1-\frac{\mi}{2} h |U_n|^2} & \frac{2 \mi}{h \lambda}
 - \lambda \left( \frac{1 + \frac{\mi}{2} h |U_n|^2}{1-\frac{\mi}{2} h |U_n|^2} \right) \\
\end{pmatrix},  \label{LP1a} \\
& \frac{d}{dt} \Phi_n(\lambda) = P_n(\lambda) \Phi_n(\lambda), \quad
P_n(\lambda) = \frac{\mi}{2} \begin{pmatrix}
 \lam ^2-|R_n|^2 &  \lam
   R_n- Q_n\lam^{-1} \\
 \lam \bar{R}_n - \bar{Q}_n \lam^{-1} & |Q_n|^2- \lam^{-2}
\end{pmatrix}, \label{LP1b}
\end{align} \label{LP1}
\end{subequations}
\hspace{-0.25cm} where $\Phi_n(\lambda) \in \mathbb{C}^2$ is defined for $n \in \mathbb{Z}$ and $\lam$ is a spectral parameter.

Because the passage from the discrete system (\ref{MTM-discrete}) to the continuum limit (\ref{MTM})
involves the change of the coordinates $U(x,t) = u(x,t-x)$ and $Q(x,t) = v(x,t-x)$,
the initial-value problem for the semi-discrete MTM system (\ref{MTM-discrete}) does not represent the initial-value problem
for the continuous MTM system (\ref{MTM}) in time variable $t$. In addition, 
numerical explorations of the semi-discrete system (\ref{MTM-discrete}) are challenging
because the last two constraints in the system (\ref{MTM-discrete})
may lead to appearance of bounded but non-decaying sequences $\{ R_n \}_{n \in \mathbb{Z}}$ and $\{ Q_n \}_{n \in \mathbb{Z}}$
in response to the bounded and decaying sequence $\{ U_n \}_{n \in \mathbb{Z}}$. On the other hand, 
since the semi-discrete MTM system (\ref{MTM-discrete}) has the Lax pair of linear equations (\ref{LP1}),
it is integrable by the inverse scattering transform method which implies existence of infinitely many
conserved quantities, exact solutions, transformations between different solutions, and
reductions to other integrable equations \cite{Joshi-book}. These properties of integrable systems
were not explored for the semi-discrete MTM system (\ref{MTM-discrete}) in the previous work \cite{JoshiJPA}.

The purpose of this work is to derive the one-fold Darboux transformation between solutions of the semi-discrete MTM
system (\ref{MTM-discrete}). We employ the Darboux transformation in order to generate one-soliton
and two-soliton solutions on zero background in the exact analytical form. By looking at the continuum limit $h \to 0$,
we show that the discrete solitons share many properties with their continuous counterparts.
We also construct one-soliton solutions on a nonzero constant background.
Further properties of the model, e.g. conserved quantities and solvability of the initial-value
problem, are left for further studies.

The following theorem represents the main result of this work.

\begin{theorem}  \label{theorem-main}
Let $\Phi_n(\lam_1) = (f_n,g_n)^T$ be a nonzero solution of the Lax pair~\eref{LP1} with $\lam=\lam_1$
and $(U_n,R_n,Q_n)$ be a solution of the semi-discrete MTM system (\ref{MTM-discrete}).
Another solution of the semi-discrete MTM system (\ref{MTM-discrete}) is given by
\begin{subequations}
\begin{align}
& U^{[1]}_n = -\frac{2\,\mi (\bar{\lam}_1 |f_n|^2 + \lam_1 |g_n|^2) U_n - h |\lam_1|^2 (\lam_1 |f_n|^2 + \bar{\lam}_1 |g_n|^2) U_n +
2\,\mi (\lam_1^2 - \bar{\lam}_1^2) f_n \bar{g}_n}{2\,\mi (\lam_1 |f_n|^2 + \bar{\lam}_1 |g_n|^2) - h |\lam_1|^2 (\bar{\lam}_1 |f_n|^2 + \lam_1 |g_n|^2)  +
h (\lam_1^2 - \bar{\lam}_1^2) \bar{f}_n g_n U_n},  \label{PoTr1}  \\
& R^{[1]}_{n} = -\frac{\left( \bar{\lambda}_1 |f_{n}|^2 + \lambda_1 |g_{n}|^2 \right) R_n
+ \left(\lambda_1^2-\bar{\lambda}_1^2\right) f_{n} \bar{g}_{n}}{\lambda_1 |f_{n}|^2 + \bar{\lambda}_1 |g_{n}|^2}, \label{PoTr2}\\
& Q^{[1]}_{n} = -\frac{|\lam_1|^2 \left( \lambda_1 |f_{n}|^2+\bar{\lambda}_1 |g_{n}|^2 \right) Q_n + \left(\lambda_1^2-\bar{\lambda}_1^2\right)
f_{n} \bar{g}_{n}}{|\lam_1|^2(\bar{\lam}_1 |f_{n}|^2+\lambda_1 |g_{n}|^2)}. \label{PoTr3}
\end{align}
\label{PoTr}
\end{subequations}
\end{theorem}

Theorem \ref{theorem-main} is proven in Section \ref{sec-one-fold} using the Lax pair
(\ref{LP1}) and the dressing methods. One-soliton and two-soliton
solutions on zero background are obtained in Section \ref{sec-soliton-zero}.
One-soliton solutions on a nonzero constant background are constructed in Section 4.
Both zero and nonzero constant backgrounds are modulationally stable in the evolution
of the semi-discrete MTM system (\ref{MTM-discrete}).
A summary and further directions are discussed in Section 5.

\section{Proof of the one-fold Darboux transformation}
\label{sec-one-fold}

The one-fold Darboux transformation takes an abstract form (see, e.g., \cite{Gu}):
\begin{align}
\Phi^{[1]}(\lambda) = T(\lambda) \Phi(\lambda), \label{eigenTran}
\end{align}
where $T(\lambda)$ is the Darboux matrix, $\Phi(\lambda)$ is a solution
to the system (\ref{LP1}), whereas $\Phi^{[1]}(\lambda)$ is a solution
of the transformed system
\begin{align}
\Phi^{[1]}_{n+1}(\lambda) = N^{[1]}_n(\lambda) \Phi^{[1]}_n(\lambda), \quad
\frac{d}{dt} \Phi^{[1]}_n(\lambda) = P^{[1]}_n(\lambda) \Phi^{[1]}_n(\lambda), \label{LP2}
\end{align}
with $N^{[1]}_{n}(\lambda)$ and $P^{[1]}_{n}(\lambda)$ having the same form as $N_n(\lambda)$ and $P_n(\lambda)$
except that the potentials $\big(U_n, Q_{n}, R_n\big)$ are replaced by
$\big(U^{[1]}_n, Q^{[1]}_{n}, R^{[1]}_n\big)$. By substituting (\ref{eigenTran})
into the linear equations (\ref{LP2}) and using the linear equations
(\ref{LP1}), we obtain the following system of equations for the Darboux matrix $T(\lambda)$:
\begin{subequations}
\begin{align}
& T_{n+1}(\lambda) N_{n}(\lambda) = N_{n}^{[1]}(\lambda) T_{n}(\lambda), \label{DTinvariancea} \\
& \frac{d}{d t} T_n(\lambda) + T_{n}(\lambda) P_n(\lambda) = P_{n}^{[1]}(\lambda) T_{n}(\lambda). \label{DTinvarianceb}
\end{align}
\label{DTinvariance}
\end{subequations}
\hspace{-0.2cm} Since $N_n(\lambda)$ and $P_n(\lambda)$ in \eref{LP1} contain both the positive and negative powers of $\lam$,
we take the one-fold Darboux matrix $T(\lambda)$ in the following form (used in \cite{XuAML2017}
in the context of the semi-discrete nonlocal nonlinear Schr\"{o}dinger equation):
\begin{equation}
T_{n}(\lam;t)=\left(\begin{array}{cc}
\sum\limits_{l=-1}^1 a_{l,n}(t)\lambda^l & \sum\limits_{l=-1}^1 b_{l,n}(t)\lambda^l\\
\sum\limits_{l=-1}^1 c_{l,n}(t)\lambda^l & \sum\limits_{l=-1}^1 d_{l,n}(t)\lambda^l
\end{array}\right), \label{DarTa}
\end{equation}
where the coefficients are to be determined. Before further work, we shall simplify the Darboux matrix
in (\ref{DarTa}) by using some constraints following from the system (\ref{DTinvariance}).
Expanding Eq.~\eref{DTinvarianceb} in powers of $\lam$ and
equating the coefficients of $\lam^3$ and $\lam^{-3}$ to $0$, we verify that
\begin{align}
b_{1,n}= c_{1,n}= b_{-1,n}= c_{-1,n}=0.    \label{Exp3b}
\end{align}
Collecting coefficients of other powers of $\lam$ yields the following
system of equations:
\begin{subequations}
\begin{align}
&\lam^2:   \,\,\,   a_{1,n} R_n - b_{0,n} - d_{1,n} R^{[1]}_n =0, \label{Exp2c} \\
&\lam^2:   \,\,\,  a_{1,n} \bar{R}^{[1]}_n -c_{0,n}- d_{1,n} \bar{R}_n =0, \label{Exp2d} \\
&\lam^{-2}:   \,\,\,  a_{-1,n} Q_n + b_{0,n} - d_{-1,n} Q^{[1]}_n= 0, \label{Exp2g} \\
&\lam^{-2}:   \,\,\,   a_{-1,n} \bar{Q}^{[1]}_n +c_{0,n} - d_{-1,n} \bar{Q}_n =0, \label{Exp2h} \\
& \lam^1:  \,\,\,  a_{0,n}R_n -  d_{0,n}R^{[1]}_n =0, \label{Exp1c} \\
& \lam^1:  \,\,\,  - a_{0,n} \bar{R}^{[1]}_n + d_{0,n} \bar{R}_n = 0,  \label{Exp1d} \\
&\lam^1:   \,\,\,  \big(|R^{[1]}_n|^2 - |R_n|^2\big)  a_{1,n}+ \bar{R}_n b_{0,n}- R^{[1]}_n
   c_{0,n} -2 \mi \frac{d a_{1,n}}{dt}=0, \label{Exp1a} \\
& \lam^1:   \,\,\,  R_n c_{0,n} - \bar{R}^{[1]}_n b_{0,n}+ \big(|Q_n|^2 -
   |Q^{[1]}_n|^2\big) d_{1,n} -2\,\mi \frac{d d_{1,n}}{dt}=0,  \label{Exp1b} \\
& \lam^{-1}: \,\,\,  a_{0,n} Q_n-d_{0,n} Q^{[1]}_n=0,  \label{Exp1e} \\
& \lam^{-1}: \,\,\, a_{0,n} \bar{Q}^{[1]}_n-d_{0,n} \bar{Q}_n=0, \label{Exp1f} \\
& \lam^{-1}: \,\,\, a_{-1,n}\big(|R^{[1]}_n|^2 -|R_n|^2 ) -b_{0,n} \bar{Q}_n + c_{0,n} Q^{[1]}_n -2 \mi \frac{d a_{-1,n}}{dt} = 0, \label{Exp1k}\\
& \lam^{-1}: \,\,\,  d_{-1,n}\big(|Q_n|^2 - |Q^{[1]}_n|^2 \big)-c_{0,n} Q_n  + b_{0,n} \bar{Q}^{[1]}_n  -2 \mi \frac{d d_{-1,n}}{dt}=0, \label{Exp1l} \\
& \lam^0: \,\,\, a_{0,n}\big(|R^{[1]}_n|^2 -|R_n|^2\big) -2 \mi \frac{d a_{0,n}}{dt}=0, \label{Exp1m}\\
& \lam^0: \,\,\, d_{0,n}\big( |Q_n|^2- |Q^{[1]}_n|^2\big)-2 \mi \frac{d d_{0,n}}{dt}=0, \label{Exp1n}\\
& \lam^0: \,\,\, b_{0,n} \big(|Q_n|^2 + |R^{[1]}_n|^2\big) - a_{1,n} Q_n+a_{-1,n} R_n + d_{1,n} Q^{[1]}_n-d_{-1,n}R^{[1]}_n -2 \mi\frac{d b_{0,n}}{dt} =0,
\label{Exp1o}\\
& \lam^0: \,\,\, c_{0,n} \big(|Q^{[1]}_n|^2 + |R_n|^2\big)  - a_{1,n} \bar{Q}^{[1]}_n + a_{-1,n}\bar{R}^{[1]}_n + d_{1,n}
   \bar{Q}_n - d_{-1,n} \bar{R}_n + 2 \mi \frac{d c_{0,n}}{dt} =0.\label{Exp1p}
\end{align}
\label{Exp1}
\end{subequations}
It follows from Eqs.~\eref{Exp1c},~\eref{Exp1d},~\eref{Exp1e} and~\eref{Exp1f} that
if $(|Q^{[1]}|,|R^{[1]}|) \neq (|Q|,|R|)$, then $a_{0,n}= d_{0,n}=0$, after which Eqs.~\eref{Exp1m} and \eref{Exp1n}
are identically satisfied. Solving Eqs.~\eref{Exp2c},~\eref{Exp2d},~\eref{Exp2g}, and~\eref{Exp2h} yields
\begin{subequations}
\begin{align}
\label{tech1a}
& b_{0,n} = a_{1,n} R_n - d_{1,n} R_n^{[1]} = d_{-1,n} Q_n^{[1]} - a_{-1,n} Q_n, \\
\label{tech1b}
& c_{0,n} = a_{1,n} \bar{R}_n^{[1]} - d_{1,n} \bar{R}_n = d_{-1,n} \bar{Q}_n - a_{-1,n} \bar{Q}_n^{[1]}.
\end{align}
\label{tech1}
\end{subequations}
Plugging (\ref{tech1}) into Eqs.~\eref{Exp1a} and \eref{Exp1l} gives
\begin{align}
\frac{d a_{1,n}}{dt} = \frac{d d_{-1,n}}{dt} = 0.  \label{adconst1}
\end{align}
All constraints of the system (\ref{Exp1}) are satisfied except for Eqs. (\ref{Exp1b}), (\ref{Exp1k}), (\ref{Exp1o}), and (\ref{Exp1p}).
It is however difficult to compute relations between the new and old potentials from these four equations.
Therefore, we will obtain the relations between $(R_n,Q_n)$ and $(R_n^{[1]},Q_n^{[1]})$ by using dressing methods
from Appendix A in \cite{ChenPel}.

Expanding Eq.~\eref{DTinvariancea} in powers of $\lam$ and
equating the coefficients of $\lam^2$ and $\lam^{-2}$ to $0$, we verify that
\begin{align}
a_{1,n+1} = a_{1,n}, \quad d_{-1,n+1}= d_{-1,n}.  \label{adconst2}
\end{align}
Combining Eqs.~\eref{adconst1} and~\eref{adconst2}, we conclude that
$a_{1,n}(t)$ and $d_{-1,n}(t)$ are constants both in $t$ and $n$.
For normalization purposes, we set $a_{1,n}(t) = 1$ and $d_{-1,n}(t) = |\lam_1|^2$.
We also re-enumerate the remaining coefficients
as follows: $a_{-1,n}(t) = a_n(t) |\lam_1|^2$, $b_{0,n}(t) = b_n(t)$,
$c_{0,n}(t) = c_n(t)$, and $d_{1,n}(t) = d_n(t)$.
The Darboux matrix $T_n^{[1]}$ given previously by (\ref{DarTa})
is now rewritten in the simplified form:
\begin{align}
T_{n}(\lam) = \begin{pmatrix}
 \lam + a_n \frac{|\lam_1|^2}{\lam}  & b_n  \\
  c_n & d_n \lam  + \frac{|\lam_1|^2}{\lam} \\
\end{pmatrix}. \label{DarTb}
\end{align}
In order to determine $a_n(t)$, $b_n(t)$, $c_n(t)$, and $d_n(t)$,
we use the symmetry properties of the Lax pair~\eref{LP1}. This allows us to find simultaneously both the coefficients of $T(\lambda)$
and the transformations between the potentials $\big(U,Q, R\big)$ and $\big(U^{[1]}, Q^{[1]}, R^{[1]} \big)$.

\begin{lemma}
Let $\Phi(\lambda_1) = \big(f, g)^T$ be a nonzero solution of the Lax pair~\eref{LP1} at $\lam=\lam_1$.
Then,
\begin{equation}
\label{other-solutions}
[\Phi(\bar{\lambda}_1)]_n = \Omega_n \left( \begin{array}{c} -\bar{g}_n \\ \bar{f}_n \end{array} \right), \quad
[\Phi(-\lambda_1)]_n = (-1)^n \left( \begin{array}{c} -f_n \\ g_n \end{array} \right), \quad
[\Phi(-\bar{\lambda}_1)]_n = (-1)^n \Omega_n \left( \begin{array}{c} \bar{g}_n \\ \bar{f}_n \end{array} \right), \quad
\end{equation}
are solutions of the Lax pair~\eref{LP1} at $\lam = \bar{\lam}_1$, $\lam = -\lam_1$, and $\lam = -\bar{\lam}_1$ respectively,
where $\Omega_n(t)$ satisfies:
\begin{subequations}
\begin{align}
\label{tech-2a}
& \Omega_{n+1} = - S_n \Omega_{n}, \quad S_n := \frac{1+\frac{\mi}{2} h |U_n|^2}{1-\frac{\mi}{2} h |U_n|^2},\\
\label{tech-2b}
&  \frac{d \Omega_{n}}{dt} = M_n \Omega_{n}, \quad M_n := \frac{\mi }{2} \left(\bar{\lambda}^2_1- \bar{\lambda}^{-2}_1 +  |Q_n|^2- |R_n|^2 \right).
\end{align}
\label{tech-2}
\end{subequations}
\label{lemma-1}
\end{lemma}
\begin{proof}
It follows from (\ref{LP1a}) that components of $\Phi(\lambda_1)$ satisfy the system of difference equations:
\begin{eqnarray}
\label{old-system-1a}
\left\{ \begin{array}{l}
f_{n+1} = \left( \lambda_1 +\frac{2 \mi}{h \lambda_1} S_n \right) f_n
+ \frac{2 U_n}{1-\frac{\mi}{2} h |U_n|^2} g_n, \\
g_{n+1} = \frac{2 \bar{U}_n}{1-\frac{\mi}{2} h |U_n|^2} f_n + \left( \frac{2 \mi}{h \lambda_1}
 - \lambda_1 S_n \right) g_n,
\end{array} \right.
\end{eqnarray}
whereas components of $\Phi(\bar{\lambda}_1)$ satisfy the system of difference equations:
\begin{eqnarray}
\label{old-system-1b}
\left\{ \begin{array}{l}
\Omega_{n+1} \bar{g}_{n+1} = \left(\bar{\lambda}_1 +\frac{2 \mi}{h \bar{\lambda}_1} S_n \right) \Omega_n \bar{g}_n
- \frac{2 U_n}{1-\frac{\mi}{2} h |U_n|^2} \Omega_n \bar{f}_n, \\
\Omega_{n+1} \bar{f}_{n+1} = -\frac{2 \bar{U}_n}{1-\frac{\mi}{2} h |U_n|^2} \Omega_n \bar{g}_n + \left( \frac{2 \mi}{h \bar{\lambda}_1} - \bar{\lambda}_1 S_n \right) \Omega_n \bar{f}_n.
\end{array} \right.
\end{eqnarray}
Dividing (\ref{old-system-1b}) by $\Omega_{n+1}$
and taking the complex conjugation yields (\ref{old-system-1a}) if and only if $\Omega$ satisfies the difference equation
(\ref{tech-2a}).
Similarly, it follows from (\ref{LP1b}) that components of $\Phi(\lambda_1)$ satisfy the time evolution equations:
\begin{eqnarray}
\label{old-system-2a}
\left\{ \begin{array}{l}
\frac{df_n}{dt} = \frac{\mi}{2} \left[
(\lambda_1^2 - |R_n|^2) f_n + (\lambda_1 R_n - \lambda_1^{-1} Q_n) g_n \right], \\
\frac{dg_n}{dt} = \frac{\mi}{2} \left[ (\lambda_1 \bar{R}_n - \lambda_1^{-1} \bar{Q}_n) f_n + (-\lambda_1^{-2} + |Q_n|^2) g_n \right],
\end{array} \right.
\end{eqnarray}
whereas components of $\Phi(\bar{\lambda}_1)$  satisfy the time evolution equations:
\begin{eqnarray}
\label{old-system-2b}
\left\{ \begin{array}{l}
\frac{d \Omega_n}{dt} \bar{g}_n + \Omega_n \frac{d \bar{g}_n}{dt} = \frac{\mi}{2} \left[
(\bar{\lambda}_1^2 - |R_n|^2) \Omega_n \bar{g}_n - (\bar{\lambda}_1 R_n - \bar{\lambda}_1^{-1} Q_n) \Omega_n \bar{f}_n \right], \\
\frac{d \Omega_n}{dt} \bar{f}_n + \Omega_n \frac{d \bar{f}_n}{dt} = \frac{\mi}{2}
\left[ -(\bar{\lambda}_1 \bar{R}_n - \bar{\lambda}_1^{-1} \bar{Q}_n) \Omega_n \bar{g}_n + (-\bar{\lambda}_1^{-2} + |Q_n|^2) \Omega_n \bar{f}_n \right],
\end{array} \right.
\end{eqnarray}
Taking the complex conjugation of (\ref{old-system-2b}) yields (\ref{old-system-2a})
if and only if $\Omega$ satisfies the time evolution equation
(\ref{tech-2b}). The other two solutions in (\ref{other-solutions}) are obtained by the symmetry
of the system (\ref{LP1}) with respect to the reflection $\lam \to -\lam$.
\end{proof}

\begin{lemma}
Let $\Phi(\lambda_1) = (f,g)^T$ be in the kernel of the Darboux matrix $T(\lambda_1)$
and $\Phi(\bar{\lambda}_1) = \Omega \big(-\bar{g}, \bar{f}\big)^T$ be in the kernel of $T(\bar{\lambda}_1)$.
Then, the coefficients of $T(\lambda)$ in (\ref{DarTb}) are given by
\begin{equation}
a_n = - \frac{\bar{\Delta}_n}{\Delta_n}, \quad
b_n = - \frac{\left(\lam_1^2 -\bar{\lam}^2_1\right) f_n\bar{g}_{n} }{\Delta_n},  \quad
c_n =  \frac{\left(\lam_1^2-\bar{\lam}^2_1\right) \bar{f}_{n} g_n}{\Delta_n}, \quad
d_n = - \frac{\bar{\Delta}_n}{\Delta_n},
\label{UndetCoeff}
\end{equation}
where $\Delta_n := \bar{\lambda}_1|f_n|^2 + \lambda_1|g_n|^2$.
Furthermore, $\Phi(-\lam_1)$ and $\Phi(-\bar{\lam}_1)$ in (\ref{other-solutions}) are in the kernel of
$T(-\lambda_1)$ and $T(-\bar{\lambda}_1)$ respectively.
\label{lemma-2}
\end{lemma}

\begin{proof}
We rewrite the linear equations for $T(\lambda_1) \Phi(\lambda_1) = 0$ and
$T(\bar{\lambda}_1) \Phi(\bar{\lambda}_1) = 0$ in the following explicit form:
\begin{equation}
\left\{ \begin{array}{l}
(\lam_1 + a_n \bar{\lam}_1) f_n + b_n g_n = 0,  \\
c_n f_n + (d_n \lam_1  +  \bar{\lam}_1) g_n = 0, \\
-(\bar{\lam}_1 + a_n \lam_1) \bar{g}_n + b_n \bar{f}_n = 0,  \\
-c_n \bar{g}_n + (d_n \bar{\lam}_1  +  \lam_1) \bar{f}_n = 0,
\end{array} \right.
\label{lin-system-a-b-c-d}
\end{equation}
where the scalar factor $\Omega$ has been canceled out.
Solving the linear system (\ref{lin-system-a-b-c-d}) with Cramer's rule yields (\ref{UndetCoeff}).
Then, it follows from (\ref{DarTb}) and (\ref{UndetCoeff}) that $T_n(\lam)$ can be written in the form:
\begin{eqnarray}
T_{n}(\lam) = \frac{(\lambda^2 - \lambda_1^2) (\lambda^2 - \bar{\lambda}_1^2)}{2\lambda
\Delta_n} \hat{T}_n(\lam),
\label{DarTc}
\end{eqnarray}
where
\begin{eqnarray*}
\hat{T}_n(\lam) & = & \frac{1}{\lambda - \lambda_1}
\left( \begin{array}{c}
\bar{g}_n  \\  \bar{f}_n \end{array} \right)
\left( \begin{array}{cc}
g_n & -f_n \end{array} \right)
+ \frac{1}{\lambda + \lambda_1}
\left( \begin{array}{c}
- \bar{g}_n  \\  \bar{f}_n \end{array} \right)
\left( \begin{array}{cc}
g_n & f_n \end{array} \right) \\
& \phantom{t} &
+ \frac{1}{\lambda - \bar{\lambda}_1}
\left( \begin{array}{c}
 f_n  \\ -  g_n \end{array} \right)
\left( \begin{array}{cc}
\bar{f}_n & \bar{g}_n \end{array} \right)
+ \frac{1}{\lambda + \bar{\lambda}_1}
\left( \begin{array}{c}
 f_n  \\  g_n \end{array} \right)
\left( \begin{array}{cc}
- \bar{f}_n & \bar{g}_n \end{array} \right).
\end{eqnarray*}
It follows from (\ref{DarTc}) that
\begin{eqnarray*}
&& T_n(\lambda_1) \left( \begin{array}{c} f_n \\ g_n \end{array} \right) = \left( \begin{array}{c} 0 \\ 0 \end{array} \right), \quad
T_n(-\lambda_1) \left( \begin{array}{c} -f_n \\ g_n \end{array} \right) = \left( \begin{array}{c} 0 \\ 0 \end{array} \right), \\
&& T_n(\bar{\lambda}_1) \left( \begin{array}{c} -\bar{g}_n \\ \bar{f}_n \end{array} \right) = \left( \begin{array}{c} 0 \\ 0 \end{array} \right), \quad
T_n(-\bar{\lambda}_1) \left( \begin{array}{c} \bar{g}_n \\ \bar{f}_n \end{array} \right) = \left( \begin{array}{c} 0 \\ 0 \end{array} \right),
\end{eqnarray*}
hence $T(\pm \lam_1) \Phi(\pm \lam_1) = 0$ and $T(\pm \bar{\lam}_1) \Phi(\pm \bar{\lam}_1) = 0$.
\end{proof}

\begin{lemma} \label{L:DTPolynomial}
Let the Darboux matrix $T(\lambda)$ be in the form \eref{DarTb} with
the coefficients given by Eqs.~\eref{UndetCoeff}.
Then, the determinant of $T(\lambda)$ is given by
\begin{align}
\det\,T_{n}(\lambda) = -  \frac{(\lam^2-\lam^2_1)(\lam^2-\bar{\lam}^2_1)}{\lam^2}
\frac{\bar{\Delta}_n}{\Delta_n}.
\label{TPolynomial}
\end{align}
\end{lemma}

\begin{proof}
Expanding $\det\,T_{n}(\lambda)$ given by \eref{DarTb} yields
\begin{align}
\det\,T_{n}(\lambda) = d_n \lambda^2 + a_n d_n |\lam_1|^2 - b_n c_n +  |\lambda_1|^2 +  a_n |\lambda_1|^4 \lambda^{-2}.
\label{TPoly-prel}
\end{align}
Since $\pm \lam_1$ and $\pm \bar\lam_1$ are the roots of $\det\,T(\lambda)$,
we obtain (\ref{TPolynomial}). Alternatively, substituting (\ref{UndetCoeff})
into (\ref{TPoly-prel}) yields (\ref{TPolynomial}).
\end{proof}

For $\lam \neq \pm \lam_1$ and $\lam \neq \pm \bar{\lam}_1$,
we define
\begin{equation}
\label{adjT}
{\rm ad} T_n(\lambda) = \det T_n(\lambda) [T_n(\lambda)]^{-1} = \begin{pmatrix}
d_n \lam  +  \frac{|\lam_1|^2}{\lam}    & -b_n  \\
  -c_n & \lam + a_n \frac{|\lam_1|^2}{\lam} \\
\end{pmatrix}.
\end{equation}
and obtain ${\rm ad} T_n(\lambda)$ from (\ref{DarTb}) and (\ref{UndetCoeff}) in the form:
\begin{eqnarray}
{\rm ad} T_{n}(\lam) = \frac{(\lambda^2 - \lambda_1^2) (\lambda^2 - \bar{\lambda}_1^2)}{2\lambda \Delta_n} {\rm ad} \hat{T}_n(\lam),
\label{DarTd}
\end{eqnarray}
where
\begin{eqnarray*}
{\rm ad} \hat{T}_n(\lam) & = & \frac{1}{\lambda - \lambda_1}
\left( \begin{array}{c}
f_n \\ g_n \end{array} \right)
\left( \begin{array}{cc}
- \bar{f}_n  & \bar{g}_n \end{array} \right)
+ \frac{1}{\lambda + \lambda_1}
\left( \begin{array}{c}
f_n \\ -g_n \end{array} \right)
\left( \begin{array}{cc}
 \bar{f}_n  & \bar{g}_n \end{array} \right) \\
& \phantom{t} &
+ \frac{1}{\lambda - \bar{\lambda}_1}
\left( \begin{array}{c}
\bar{g}_n \\ -\bar{f}_n \end{array} \right)
\left( \begin{array}{cc}
- g_n  & - f_n \end{array} \right)
+ \frac{1}{\lambda + \bar{\lambda}_1}
\left( \begin{array}{c}
\bar{g}_n \\ \bar{f}_n \end{array} \right)
\left( \begin{array}{cc}
 g_n  & - f_n \end{array} \right).
\end{eqnarray*}

New potentials $N_n^{[1]}(\lam)$ and $P_n^{[1]}(\lam)$ are derived from Eqs.~(\ref{DTinvariance}) 
by using the Darboux matrix $T(\lam)$. Assuming $\lam \neq \pm \lam_1$ and $\lam \neq \pm \bar{\lam}_1$, 
we obtain from (\ref{DTinvariance}) and (\ref{DarTd}) that 
\begin{eqnarray}
\nonumber
N_n^{[1]}(\lam) & = & \frac{1}{\det T_n(\lam)} T_{n+1}(\lam) N_n(\lam) {\rm ad} T_n(\lam) \\
& = & -\frac{\lam}{2   \bar{\Delta}_n} T_{n+1}(\lam) N_n(\lam) {\rm ad} \hat{T}_n(\lam)
\label{representationa}
\end{eqnarray}
and
\begin{eqnarray}
\nonumber
P^{[1]}_n(\lambda) & = & \frac{1}{\det T_n(\lam)} \left[ \frac{d}{d t} T_n(\lambda) + T_{n}(\lambda) P_n(\lambda) \right] {\rm ad} T_{n}(\lambda) \\
& = & -\frac{\lam}{2  \bar{\Delta}_n}  \left[ \frac{d}{d t} T_n(\lambda) + T_{n}(\lambda) P_n(\lambda) \right] {\rm ad} \hat{T}_{n}(\lambda),
\label{representationb}
\end{eqnarray}
where the expressions (\ref{TPolynomial}) and (\ref{DarTd}) have been used.

First, we compute the products in the right-hand side of Eq.~\eref{representationa}. By Lemma \ref{lemma-1}
and direct computations, we obtain
\begin{subequations}
\label{Nnf}
\begin{align}
 & N_n(\lam) \left( \begin{array}{c}
f_n \\ g_n \end{array} \right)  =   \left( \begin{array}{c}
f_{n+1} \\ g_{n+1} \end{array} \right) + (\lam - \lam_1)
\left( \begin{array}{cc} 1 - \frac{2\,\mi}{h \lam \lam_1} S_n & 0 \\
0 &  -\frac{2\,\mi}{h \lam \lam_1} - S_n \end{array} \right)
\left( \begin{array}{c}  f_{n} \\ g_{n} \end{array} \right), \\
&  N_n(\lam) \left( \begin{array}{c}
f_n \\ -g_n \end{array} \right)   =  \left( \begin{array}{c}
-f_{n+1} \\ g_{n+1} \end{array} \right) + (\lam + \lam_1)
\left( \begin{array}{cc} 1 + \frac{2\,\mi}{h \lam \lam_1} S_n & 0 \\
0 &  \frac{2\,\mi}{h \lam \lam_1} - S_n \end{array} \right)
\left( \begin{array}{c}  f_{n} \\ -g_{n} \end{array} \right), \\
& N_n(\lam) \left( \begin{array}{c}
\bar{g}_n \\ -\bar{f}_n \end{array} \right)   =
-S_n \left( \begin{array}{c}
\bar{g}_{n+1} \\ -\bar{f}_{n+1} \end{array} \right) + (\lam - \bar{\lam}_1)
\left( \begin{array}{cc} 1 - \frac{2\,\mi}{h \lam \bar{\lam}_1} S_n & 0 \\
0 &  -\frac{2\,\mi}{h \lam \bar{\lam}_1} - S_n \end{array} \right)
\left( \begin{array}{c} \bar{g}_n \\ -\bar{f}_n \end{array} \right),  \\
& N_n(\lam) \left( \begin{array}{c}
\bar{g}_n \\ \bar{f}_n \end{array} \right)   =
S_n \left( \begin{array}{c}
\bar{g}_{n+1} \\ \bar{f}_{n+1} \end{array} \right) + (\lam + \bar{\lam}_1)
\left( \begin{array}{cc} 1 + \frac{2\,\mi}{h \lam \bar{\lam}_1} S_n & 0 \\
0 &  \frac{2\,\mi}{h \lam \bar{\lam}_1} - S_n \end{array} \right)
\left( \begin{array}{c}  \bar{g}_n \\ \bar{f}_n \end{array} \right),
\end{align}
\end{subequations}
where $S_n$ is defined in Eq.~(\ref{tech-2a}). By using this table,
we compute the first product in (\ref{representationa}):
\begin{eqnarray*}
N_n(\lam) {\rm ad} \hat{T}_n(\lam) & = & \frac{1}{\lambda - \lambda_1}
\left( \begin{array}{c}
f_{n+1} \\ g_{n+1} \end{array} \right)
\left( \begin{array}{cc}
- \bar{f}_n  & \bar{g}_n \end{array} \right)
+ \frac{1}{\lambda + \lambda_1}
\left( \begin{array}{c}
-f_{n+1} \\ g_{n+1} \end{array} \right)
\left( \begin{array}{cc}
 \bar{f}_n  & \bar{g}_n \end{array} \right) \\
& \phantom{t} &
+ \frac{1}{\lambda - \bar{\lambda}_1}
S_n \left( \begin{array}{c}
\bar{g}_{n+1} \\ -\bar{f}_{n+1} \end{array} \right)
\left( \begin{array}{cc}
 g_n  & f_n \end{array} \right) + \frac{1}{\lambda + \bar{\lambda}_1} S_n
\left( \begin{array}{c}
\bar{g}_{n+1} \\ \bar{f}_{n+1} \end{array} \right)
\left( \begin{array}{cc}
 g_n  & - f_n \end{array} \right) \\
& \phantom{t} &
+ \frac{4\mi}{h \lam |\lam_1|^2}
\left( \begin{array}{cc}
 S_n \Delta_n & 0 \\
0 & -\bar{\Delta}_n \end{array} \right).
\end{eqnarray*}
By Lemma \ref{lemma-2} and direct computations, we obtain
\begin{subequations}
\label{Tnf}
\begin{align}
& T_n(\lam) \left( \begin{array}{c}
f_n \\ g_n \end{array} \right) =  (\lam - \lam_1)
\left( \begin{array}{cc} 1 - a_n \frac{\bar{\lam}_1}{\lam} & 0 \\
0 &  d_n -  \frac{\bar{\lam}_1}{\lam } \end{array} \right)
\left( \begin{array}{c}  f_n \\ g_n \end{array} \right), \\
& T_n(\lam) \left( \begin{array}{c}
f_n \\ -g_n \end{array} \right) = (\lam + \lam_1)
\left( \begin{array}{cc} 1 + a_n \frac{\bar{\lam}_1}{\lam} & 0 \\
0 &  d_n +  \frac{\bar{\lam}_1}{\lam } \end{array} \right)
\left( \begin{array}{c}  f_n \\ -g_n \end{array} \right), \\
& T_n(\lam) \left( \begin{array}{c}
\bar{g}_n \\ -\bar{f}_n \end{array} \right) = (\lam - \bar{\lam}_1)
\left( \begin{array}{cc} 1 - a_n \frac{\lam_1}{\lam} & 0 \\
0 &  d_n -  \frac{\lam_1}{\lam} \end{array} \right)
\left( \begin{array}{c}  \bar{g}_n \\ -\bar{f}_{n} \end{array} \right), \\
& T_n(\lam) \left( \begin{array}{c}
\bar{g}_n \\ \bar{f}_n \end{array} \right) = (\lam + \bar{\lam}_1)
\left( \begin{array}{cc} 1 + a_n \frac{\lam_1}{\lam} & 0 \\
0 &  d_n +  \frac{\lam_1}{\lam} \end{array} \right)
\left( \begin{array}{c}  \bar{g}_n \\ \bar{f}_{n} \end{array} \right).
\end{align}
\end{subequations}
By using this table, we compute the second product in (\ref{representationa}):
\begin{eqnarray*}
T_{n+1}(\lam) N_n(\lam) {\rm ad} \hat{T}_n(\lam) & = &
2 \left( \begin{array}{cc} - \left( f_{n+1} \bar{f}_n - S_n \bar{g}_{n+1} g_n \right) &
- \frac{a_{n+1}}{\lam} \left( \bar{\lam}_1 f_{n+1} \bar{g}_n + S_n \lam_1 \bar{g}_{n+1} f_n \right) \\
\frac{1}{\lam} \left( \bar{\lam}_1 g_{n+1} \bar{f}_n + S_n \lam_1 \bar{f}_{n+1} g_n \right) &
d_{n+1} \left( g_{n+1} \bar{g}_n - S_n \bar{f}_{n+1} f_n \right)
\end{array} \right) \\
& \phantom{t} &
+ \frac{4\mi}{h \lam |\lam_1|^2} \left( \begin{array}{cc} \lam + a_{n+1} \frac{|\lam_1|^2}{\lam}  & b_{n+1}  \\
  c_{n+1} & d_{n+1} \lam  +  \frac{|\lam_1|^2}{\lam}  \end{array} \right)
\left( \begin{array}{cc}  S_n \Delta_n & 0 \\
0 & -\bar{\Delta}_n \end{array} \right).
\end{eqnarray*}
Substituting this expression into (\ref{representationa}), we finally obtain
\begin{eqnarray}
\label{N-new}
N_n^{[1]}(\lam) = \left( \begin{array}{cc} \delta_0 \lam + \frac{2\,\mi}{h \lam} \delta_1 & \delta_2 \\
\delta_3 & \frac{2\,\mi}{h\lam} - \delta_4 \lam  \end{array} \right),
\end{eqnarray}
where
\begin{eqnarray*}
\delta_0 & = &  \frac{\bar{f}_n f_{n+1} - S_n g_n \bar{g}_{n+1}}{\bar{\Delta}_n} - \frac{2\,\mi}{h} \frac{S_n \Delta_n}{|\lambda_1|^2 \bar{\Delta}_n},  \\
\delta_1 & = & -\frac{a_{n+1} S_n \Delta_n}{\bar{\Delta}_n}, \\
\delta_2 & = & a_{n+1} \frac{\bar{\lam}_1 f_{n+1} \bar{g}_n + S_n \lam_1 \bar{g}_{n+1} f_n}{
 \bar{\Delta}_n} + \frac{2\,\mi b_{n+1}}{h  |\lam_1|^2}, \\
\delta_3 & = & - \frac{\bar{\lam}_1 g_{n+1} \bar{f}_n + S_n \lam_1 \bar{f}_{n+1} g_n}{
\bar{\Delta}_n} - \frac{2\,\mi c_{n+1} S_n \Delta_n}{h |\lam_1|^2 \bar{\Delta}_n}, \\
\delta_4 & = & -\frac{2\,\mi d_{n+1}}{h  |\lam_1|^2} +
d_{n+1} \frac{g_{n+1} \bar{g}_n - S_n \bar{f}_{n+1} f_n}{ \bar{\Delta}_n}.
\end{eqnarray*}
It follows from substitution of~(\ref{old-system-1a}) and~(\ref{old-system-1b}) for $f_{n+1}$, $g_{n+1}$, $\bar{f}_{n+1}$ and $\bar{g}_{n+1}$ that
\begin{align*}
\bar{f}_n f_{n+1} - S_n g_n \bar{g}_{n+1} = \bar{\Delta}_n + \frac{2\,\mi S_n \Delta_n}{h|\lam_1|^2}
\end{align*}
and
\begin{align*}
g_{n+1} \bar{g}_n - S_n \bar{f}_{n+1} f_n = - S_n \Delta_n  + \frac{2\,\mi \bar{\Delta}_n}{h|\lam_1|^2}.
\end{align*}
As a result, we verify that $\delta_0 = 1$ and $\delta_1 = \delta_4$.
We represent $N_n^{[1]}(\lam)$ in (\ref{N-new}) in the same form as
$N_n(\lam)$ in (\ref{LP1a}), therefore, we write
\begin{align}
 \label{WYZ1}
 \delta_1 =\frac{1 + \frac{\mi}{2}h  W_n}{1-\frac{\mi}{2}h  W_n},  \quad
 \delta_2  = \frac{2 Y_n}{1-\frac{\mi}{2}h W_n}, \quad
 \delta_3   = \frac{2 Z_n}{1-\frac{\mi}{2}h W_n}
\end{align}
for some $Y_n$, $Z_n$, and $W_n$. Using Eqs.~(\ref{UndetCoeff}) for $a_{n+1}$, $b_{n+1}$, and $c_{n+1}$
and solving Eq.~(\ref{WYZ1}) for $W_n$, $Y_n$, and $Z_n$ yield
\begin{subequations}
\label{WYZ2}
\begin{align}
& W_n = \frac{2\,\mi (\bar{\Delta}_n \Delta_{n+1} - S_n \bar{\Delta}_{n+1}\Delta_n)}{
h(\bar{\Delta}_n \Delta_{n+1} + S_n \bar{\Delta}_{n+1}\Delta_n)},  \label{WYZ2a} \\
& Y_n = - \frac{h|\lambda_1|^2 \bar{\Delta}_{n+1} \left(\lambda_1 S_n f_n \bar{g}_{n+1} + \bar{\lambda}_1 f_{n+1} \bar{g}_n\right)
 + 2\,\mi (\lambda_1^2-\bar{\lambda}_1^2)\bar{\Delta}_n f_{n+1} \bar{g}_{n+1}}{
h  |\lambda_1|^2 ( \bar{\Delta}_n \Delta_{n+1} + S_n \bar{\Delta}_{n+1}\Delta_n)}, \label{WYZ2b} \\
& Z_n = - \frac{h|\lambda_1|^2 \Delta_{n+1}
\left(\lambda_1 S_n \bar{f}_{n+1}g_n + \bar{\lambda}_1\bar{f}_n g_{n+1} \right)
+ 2\,\mi (\lambda_1^2-\bar{\lambda}_1^2) S_n \Delta_n \bar{f}_{n+1}g_{n+1} }{
h |\lambda_1|^2 (\bar{\Delta}_n \Delta_{n+1} + S_n \bar{\Delta}_{n+1}\Delta_n)}.\label{WYZ2c}
\end{align}
\end{subequations}
Substituting Eqs.~(\ref{old-system-1a}) and~(\ref{old-system-1b}) into Eqs.~(\ref{WYZ2b})--(\ref{WYZ2c})
simplifies $Y_n$ and $Z_n$ to the form:
\begin{subequations}
\label{Y-Z-final}
\begin{align}
& Y_n = \frac{h |\lambda_1 |^2 \bar{\Delta}_n U_n - 2 \mi \left(\lambda_1^2-\bar{\lambda}_1^2\right) f_n \bar{g}_n-2 \mi \Delta_n U_n}{h \left(\lambda_1^2-\bar{\lambda}_1^2\right)\bar{f}_n g_n U_n -h |\lambda_1 |^2 \Delta_n + 2 \mi \bar{\Delta}_n}, \label{Y-Z-finalA} \\
& Z_n = \frac{h |\lambda_1 |^2 \Delta_n \bar{U}_n-2 \mi
\left(\lambda_1^2-\bar{\lambda}_1^2\right) g_n \bar{f}_n+2 \mi \bar{\Delta}_n\bar{U}_n}{h \left(\bar{\lambda}_1^2-\lambda_1^2\right) f_n  \bar{g}_n \bar{U}_n - h |\lambda_1 |^2 \bar{\Delta}_n  - 2 \mi \Delta_n}. \label{Y-Z-finalB}
\end{align}
\end{subequations}
It follows from Eqs.~(\ref{Y-Z-final}) that $Y_n = \bar{Z}_n$.
We have checked with the aid of Wolfram's MATHEMATICA from Eq.~(\ref{WYZ2a}) that
$W_n =Y_n Z_n $ is satisfied. As a result, we conclude that
$N_n^{[1]}(\lam)$ in (\ref{N-new}) is the same as that of $N_n(\lam)$ in (\ref{LP1a})
with the correspondence: $U_n^{[1]}=Y_n$, $\overline{U_n^{[1]}} = Z_n = \bar{Y}_n$,
and $|U_n^{[1]}|^2 = W_n = |Y_n|^2$. Thus, Eq.~(\ref{PoTr1}) follows from
the transformation formula (\ref{Y-Z-finalA}).

Next, we prove Eq.~\eref{representationb} and derive the transformations for $R_n$ and $Q_n$
in Eqs. (\ref{PoTr2}) and (\ref{PoTr3}). Again, using Lemma \ref{lemma-1} and direct computations, we
obtain
\begin{subequations}
\label{Pnf}
\begin{align}
& P_n(\lam) \left(\begin{array}{c} f_n \\ g_n \end{array} \right)
=  \left(\begin{array}{c} f_{n,t} \\ g_{n,t} \end{array} \right)
+ (\lam - \lam_1) H_1(\lam) \left(\begin{array}{c} f_n \\ g_n \end{array} \right),  \\
& P_n(\lam) \left( \begin{array}{c} f_n \\ -g_n \end{array} \right)
=  \left( \begin{array}{c} f_{n,t} \\ -g_{n,t} \end{array} \right)
+ (\lam + \lam_1) H_2(\lam)\left( \begin{array}{c} f_n \\ -g_n \end{array} \right),     \\
&  P_n(\lam) \left( \begin{array}{c} \bar{g}_n \\[1mm] -\bar{f}_n \end{array} \right)
= \left( \begin{array}{c} \bar{g}_{n,t} \\[1mm] -\bar{f}_{n,t} \end{array}\right)
+ M_n \left( \begin{array}{c} \bar{g}_n \\[1mm] -\bar{f}_n \end{array} \right)
+ (\lam - \bar{\lam}_1) H_3(\lam) \left( \begin{array}{c} \bar{g}_n \\[1mm] -\bar{f}_n \end{array} \right),    \\
& P_n(\lam) \left( \begin{array}{c} \bar{g}_n \\[1mm] \bar{f}_n \end{array} \right)
=  \left( \begin{array}{c} \bar{g}_{n,t} \\[1mm] \bar{f}_{n,t} \end{array} \right)
+ M_n \left( \begin{array}{c} \bar{g}_n \\[1mm] \bar{f}_n \end{array} \right)
+ (\lam + \bar{\lam}_1) H_4(\lam) \left( \begin{array}{c} \bar{g}_n \\[1mm] \bar{f}_n \end{array} \right),
\end{align}
\end{subequations}
where $M_n$ is defined in Eq.~(\ref{tech-2b}) and matrices $H_{1,2,3,4}(\lam)$ are given by
\begin{align*}
& H_1(\lam) = \frac{\mi}{2}  \left( \begin{array}{cc}
\lambda + \lambda_1 & R_n + \frac{1}{\lambda \lambda_1} Q_n \\[1mm]
\bar{R}_n + \frac{1}{\lambda \lambda_1} \bar{Q}_n & \frac{\lambda + \lambda_1}{\lambda^2 \lambda^2_1} \end{array} \right), \\[1mm]
& H_2(\lam) =  \frac{\mi}{2} \left( \begin{array}{cc}
\lambda -\lambda_1 & R_n - \frac{1}{\lam \lam_1} Q_n \\[1mm]
\bar{R}_n - \frac{1}{\lam \lam_1} \bar{Q}_n & \frac{\lambda - \lambda_1}{\lambda^2 \lambda^2_1}  \end{array} \right), \\[1mm]
& H_3(\lam) =  \frac{\mi}{2} \left( \begin{array}{cc}
\lambda + \bar{\lambda}_1 & R_n + \frac{1}{\lambda \bar{\lambda}_1} Q_n  \\[1mm]
\bar{R}_n + \frac{1}{\lambda \bar{\lambda}_1} \bar{Q}_n  & \frac{\lam + \bar{\lam}_1}{\lambda^2 \bar{\lambda}_1^2} \end{array} \right), \\[1mm]
& H_4(\lam)=  \frac{\mi}{2} \left( \begin{array}{cc}
\lambda - \bar{\lambda}_1 &  R_n - \frac{1}{\lambda \bar{\lambda}_1} Q_n  \\[1mm]
\bar{R}_n -  \frac{1}{\lambda \bar{\lambda}_1} \bar{Q}_n & \frac{\lam - \bar{\lam}_1}{\lambda^2 \bar{\lambda}_1^2} \end{array} \right).
\end{align*}
Based on the results in Eq.~\eref{Pnf}, the product in the right-hand side of Eq.~\eref{representationb} can be obtained as
\begin{eqnarray}
\nonumber &  &
\left[ \frac{d} {d t} T_n(\lambda) + T_{n}(\lambda) P_n(\lambda) \right] {\rm ad} \hat{T}_{n}(\lambda) \\
\nonumber &  &  =
\frac{1}{\lambda -\lambda_1}
\left[T_n(\lam)\left(
\begin{array}{c}
 f_n \\
 g_n \\
\end{array}
\right) \right]_t
\left(
\begin{array}{cc}
 - \bar{f}_n & \bar{g}_n \\
\end{array}
\right) +
\frac{1}{\lambda +\lambda_1} \left[T_n(\lam)\left(
\begin{array}{c}
 f_n \\
 -g_n \\
\end{array}
\right) \right]_t \left(
\begin{array}{cc}
  \bar{f}_n & \bar{g}_n \\
\end{array}
\right) \\[1mm]
\nonumber &  & \quad +\frac{1}{\lambda -\bar{\lambda}_1}
\left[
T_n(\lam)\left(
\begin{array}{c}
 \bar{g}_n \\
 -\bar{f}_n \\
\end{array}
\right) \right]_t \left(
\begin{array}{cc}
 - g_n & - f_n \\
\end{array}
\right)
+\frac{1}{\lambda+ {\bar{\lambda}_1}}
\left[T_n(\lam)\left(
\begin{array}{c}
 \bar{g}_n \\
 \bar{f}_n \\
\end{array}
\right)\right]_t\left(
\begin{array}{cc}
  g_n & -f_n \\
\end{array}
\right)
\\
\nonumber &  & \quad + T_{n}(\lam)H_1(\lam)
\left( \begin{array}{c}
f_{n} \\ g_{n} \end{array} \right)
\left( \begin{array}{cc}
- \bar{f}_n  & \bar{g}_n \end{array} \right)
+ T_{n}(\lam)H_2(\lam)
\left( \begin{array}{c}
f_{n} \\ -g_{n} \end{array} \right)
\left( \begin{array}{cc}
 \bar{f}_n  & \bar{g}_n \end{array} \right) \\
\nonumber & & \quad
+ T_{n}(\lam)H_3(\lam)
\left( \begin{array}{c}
\bar{g}_{n} \\ -\bar{f}_{n} \end{array} \right)
\left( \begin{array}{cc}
- g_n  & - f_n \end{array} \right)
+ T_{n}(\lam)H_4(\lam)
\left( \begin{array}{c}
\bar{g}_{n} \\ \bar{f}_{n} \end{array} \right)
\left( \begin{array}{cc}
 g_n  & - f_n \end{array} \right) \\
\nonumber & & \quad + M_n \left[ \frac{1}{\lambda - \bar{\lambda}_1} T_{n}(\lam) \left( \begin{array}{c}
\bar{g}_{n} \\ -\bar{f}_{n} \end{array} \right)
\left( \begin{array}{cc}
- g_n  & - f_n \end{array} \right)
+ \frac{1}{\lambda + \bar{\lambda}_1} T_{n}(\lam)
\left( \begin{array}{c}
\bar{g}_{n} \\ \bar{f}_{n} \end{array} \right)
\left( \begin{array}{cc}
 g_n  & - f_n \end{array} \right)\right].
\end{eqnarray}
Expanding the above equation and substituting it into ~(\ref{representationb}) gives
\begin{align}
\nonumber
P^{[1]}_n(\lambda)
= & \frac{1}{\bar{\Delta}_n} \left(
\begin{array}{cc}
- \bar{\lambda}_1 \bar{f}_n \left(a_n f_n\right)_t
- \lambda_1 g_n \left(a_n \bar{g}_n\right)_t &
\lambda \left(f_n\bar{g}_{n,t} - f_{n,t}\bar{g}_n\right) \\
\lambda\, d_n \left(\bar{f}_n g_{n,t} - \bar{f}_{n,t} g_n\right)
 & \lambda_1 f_n \bar{f}_{n,t} + \bar{\lambda}_1 \bar{g}_n g_{n,t} \\
\end{array}
\right) \\
& \nonumber
+ \frac{\mi}{2} \left(
\begin{array}{cc}
\lambda^2 + |\lambda_1|^2 a_n
+ \frac{b_n}{|\lambda_1|^2}\left(\frac{\Delta_n}{\bar{\Delta}_n}\bar{Q}_n
-\frac{\lambda_1^2-\bar{\lambda}_1^2}{|\lambda_1|^2 \bar{\Delta}_n}\bar{f}_n g_n \right) &
- \left(\frac{a_n}{\lambda} +\frac{\lambda}{|\lambda_1|^2}\right)Q_n
-\frac{b_n}{\lambda |\lambda_1|^2} \\[2mm]
\lambda c_n +\left(\frac{1}{\lambda} + \frac{\lambda  d_n}{|\lambda_1|^2} \right)
 \left(\frac{\Delta_n}{\bar{\Delta}_n}\bar{Q}_n-\frac{\lambda_1^2-\bar{\lambda}_1^2 }
 {|\lambda_1|^2 \bar{\Delta}_n } g_n \bar{f}_n  \right) &
- \frac{1}{\lambda^2} - \frac{(d_n+c_n Q_n)}{ |\lambda_1|^2}
\end{array}
\right)  \\
& +
M_n
\left(
\begin{array}{cc}
\frac{\lambda_1}{\Delta_n}|g_n|^2 &  \frac{\lam}{\bar{\Delta}_n} f_n \bar{g}_n \\[1mm]
\frac{\lambda}{\Delta_n}\bar{f}_n g_n &
\frac{\lambda_1}{\bar{\Delta}_n}|f_n|^2 \\
\end{array}
\right),   \label{PN1a}
\end{align}
where we have used Eq.~(\ref{UndetCoeff}) in obtaining the last term. Thus, $P^{[1]}_n$ can be formally written in the form
\begin{align}
P^{[1]}_n(\lambda)
= \frac{\mi}{2}\left( \begin{array}{cc}  \lam^2 - \pi_1 \pi_3 & \pi_1 \lam - \pi_2 \lam^{-1} \\[1mm]
\pi_3 \lam - \pi_4 \lam^{-1}  & \pi_2 \pi_4 - \lam^{-2}   \end{array} \right),   \label{PN1b}
\end{align}
Comparing Eqs.~(\ref{PN1a}) and~(\ref{PN1b}) and using Eqs.~\eref{UndetCoeff} together with~\eref{old-system-2a},
we can express $\pi_i$'s ($1\leq i\leq 4$) as
\begin{subequations}
\label{Pi-equations}
\begin{align}
& \pi_1 = - \frac{\Delta_n R_n + \left(\lambda_1^2-\bar{\lambda}_1^2\right)f_{n}\bar{g}_{n}}
   { \bar\Delta_n},  \label{Pi-equations-A} \\[1mm]
&  \pi_2= - \frac{|\lam_1|^2 \bar\Delta_n Q_n + \left(\lambda_1^2-\bar{\lambda
  }_1^2\right)f_{n} \bar{g}_{n} }
   {|\lam_1|^2 \Delta_n},  \label{Pi-equations-B} \\[1mm]
& \pi_3 = - \frac{\bar\Delta_n\bar{R}_n
+ \left(\bar{\lambda}^2_1-\lambda^2_1\right)   \bar{f}_{n}g_{n}}
{\Delta_n}, \label{Pi-equations-C} \\[1mm]
&  \pi_4 = - \frac{|\lambda_1|^2\Delta_n \bar{Q}_n + \left(\bar{\lambda}^2_1-\lambda^2_1\right)\bar{f}_{n} g_{n} }{|\lambda_1|^2 \bar\Delta_n},
\label{Pi-equations-D}
\end{align}
\end{subequations}
where Wolfram's MATHEMATICA has been used for simplification. It is obvious from (\ref{Pi-equations})
that $\bar{\pi}_1 = \pi_3$ and $\bar{\pi}_2 = \pi_4$. As a result, we conclude that
$P_n^{[1]}(\lam)$ in (\ref{PN1a}) is the same as that of $P_n(\lam)$ in (\ref{LP1b})
with the correspondence: $R_n^{[1]}=\pi_1$ and $Q_n^{[1]}=\pi_2$. Thus, Eqs.~(\ref{PoTr2})--(\ref{PoTr3}) follow from
the transformation formulas (\ref{Pi-equations-A})--(\ref{Pi-equations-B}). 
Theorem \ref{theorem-main} is proven with the algorithmic computations.

\section{Soliton solutions on zero background}
\label{sec-soliton-zero}

Here we use the one-fold Darboux transformation of Theorem \ref{theorem-main} and
construct soliton solutions on zero background. Hence we take zero potentials
$(U,R,Q) = (0,0,0)$ in the transformation formula (\ref{PoTr}) and obtain
\begin{subequations}
\begin{align}
& U^{[1]}_n = -\frac{2\,\mi (\lam_1^2 - \bar{\lam}_1^2) f_n \bar{g}_n}{2\,\mi (\lam_1 |f_n|^2 + \bar{\lam}_1 |g_n|^2) -
h |\lam_1|^2 (\bar{\lam}_1 |f_n|^2 + \lam_1 |g_n|^2)},  \label{PoTr1a}  \\
& R^{[1]}_{n} = -\frac{(\lambda_1^2-\bar{\lambda}_1^2) f_{n} \bar{g}_{n}}{\lambda_1 |f_{n}|^2 + \bar{\lambda}_1 |g_{n}|^2}, \label{PoTr2a}\\
& Q^{[1]}_{n} = -\frac{(\lambda_1^2-\bar{\lambda}_1^2) f_{n} \bar{g}_{n}}{|\lam_1|^2(\bar{\lam}_1 |f_{n}|^2+\lambda_1 |g_{n}|^2)}, \label{PoTr3a}
\end{align}
\label{PoTra}
\end{subequations}
where $\Phi_n(\lam_1) = (f_n,g_n)^T$ is a nonzero solution of the Lax pair (\ref{LP1}) with $\lam = \lam_1$ at
the zero background. First, we prove that the zero background is linearly stable in the semi-discrete MTM system (\ref{MTM-discrete}).
Next, we construct Jost solutions of the Lax pair (\ref{LP1}) at the zero background.
At last, we obtain and study the exact expressions for one-soliton and two-soliton solutions.

\subsection{Stability of zero background}

Linearization of the semi-discrete MTM system (\ref{MTM-discrete}) at the zero background is written
as the linear system
\begin{equation}
\label{MTM-linear}
\left\{ \begin{array}{l}
\displaystyle
4 \mi \frac{d u_n}{dt} + q_{n+1} + q_n  + \frac{2 \mi}{h} (r_{n+1}-r_n) = 0, \\
\displaystyle
q_{n+1} - q_n + \mi h u_n = 0, \\
\displaystyle
r_{n+1} + r_n - 2 u_n = 0. \end{array} \right.
\end{equation}
Thanks to the linear superposition principle, we use the discrete
Fourier transform on the lattice,
\begin{equation}
\label{Fourier}
u_n = \frac{1}{2\pi} \int_{-\pi}^{\pi} \hat{u}(\theta) e^{\mi n \theta} d \theta, \quad n \in \mathbb{Z},
\end{equation}
invert the second and third equations of the differential-difference system (\ref{MTM-linear}),
and obtain the following differential equation with parameter $\theta \in (-\pi,\pi)\backslash \{0\}$:
\begin{equation}
\label{closed-eq-zero}
h \frac{d \hat{u}}{d t} = \left( \frac{h^2}{4} \frac{e^{\mi \theta} + 1}{e^{\mi \theta} - 1} - \frac{e^{\mi \theta} -1}{e^{\mi \theta} +1} \right) \hat{u}.
\end{equation}
Separating variables in $\hat{u} = \hat{u}_0(\theta) e^{-\mi t \omega(\theta)}$ %with $\omega(\theta)$
yields the dispersion relation for the Fourier mode $\hat{u}_0(\theta)$:
\begin{equation}
\omega(\theta) = \frac{1}{h \sin \theta} \left[ \left( \frac{h^2}{4} + 1 \right) + \left( \frac{h^2}{4} - 1 \right) \cos \theta \right], \quad \theta \in (-\pi,\pi) \backslash \{0\}.
\label{dispersion}
\end{equation}
Since $\omega(\theta) \in \mathbb{R}$ for every $\theta \in (-\pi,\pi) \backslash \{0\}$, the zero background is
linearly stable. Note however that $|\omega(\theta)| \to \infty$ as $\theta \to 0$ and $\theta \to \pm \pi$.
Divergences of the dispersion relation in (\ref{dispersion}) as $\theta \to 0$ and $\theta \to \pm \pi$ are
related to inversion of the second and third difference equations in the linear system (\ref{MTM-linear}).

\subsection{Solutions of the Lax pair (\ref{LP1}) at zero background}

Lax pair~\eref{LP1} at the zero background is decoupled into two systems which admit the following
two linearly independent solutions:
\begin{equation}
\label{Jost}
[\Phi_+(\lam)]_n(t) = \alpha \left( \begin{array}{c} 1 \\
0 \end{array} \right) \mu_+^n e^{\frac{\mi \lambda^2}{2}  t}, \quad
[\Phi_-(\lam)]_n(t) = \beta \left( \begin{array}{c} 0 \\
1 \end{array} \right) \mu_-^n e^{-\frac{\mi}{2 \lambda^2}t},
\end{equation}
where $\alpha,\beta \in \mathbb{C}\backslash\{0\}$ are parameters and
$$
\mu_{\pm}(\lam) := \frac{2\,\mi}{h \lam} \pm \lam.
$$
We say that $\Phi(\lam)$ is the Jost function if $\lam \in \mathbb{C}$ yields either $|\mu_+(\lam)| = 1$ or $|\mu_-(\lam)| = 1$, in which case
one of the two fundamental solutions in (\ref{Jost})
is bounded in the limit $|n| \to \infty$. Constraints $|\mu_{\pm}(\lam)| = 1$
for $\lambda = |\lambda| e^{\mi \theta/2}$ in the polar form
are equivalent to the following equation:
\begin{equation}
\label{transc}
|\lambda|^2 \pm \frac{4}{h} \sin(\theta) + \frac{4}{h^2 |\lambda|^2} = 1.
\end{equation}
Roots of Eq.~(\ref{transc}) in the complex plane for $\lambda \in \mathbb{C}$ are
shown on Fig. \ref{fig-plane} for $h < 4$ (left) and $h > 4$ (right).
For every $\lambda$ on each curve of the Lax spectrum,
there exists one Jost function in (\ref{Jost}) which remains bounded in the limit $|n| \to \infty$.
On the other hand, thanks to the time dependence in (\ref{Jost}),
Jost functions remain bounded also in the limit $|t| \to \infty$
if and only if $\lambda^2 \in \mathbb{R}$.
No such Jost functions exist for $h < 4$ as is seen from the left panel of Fig. \ref{fig-plane}.
In other words, all Jost functions diverge exponentially either as $t \to -\infty$ or
as $t \to +\infty$ if $h < 4$.

\begin{figure}[h!]
 \centering
\includegraphics[width = 3in, height = 2in]{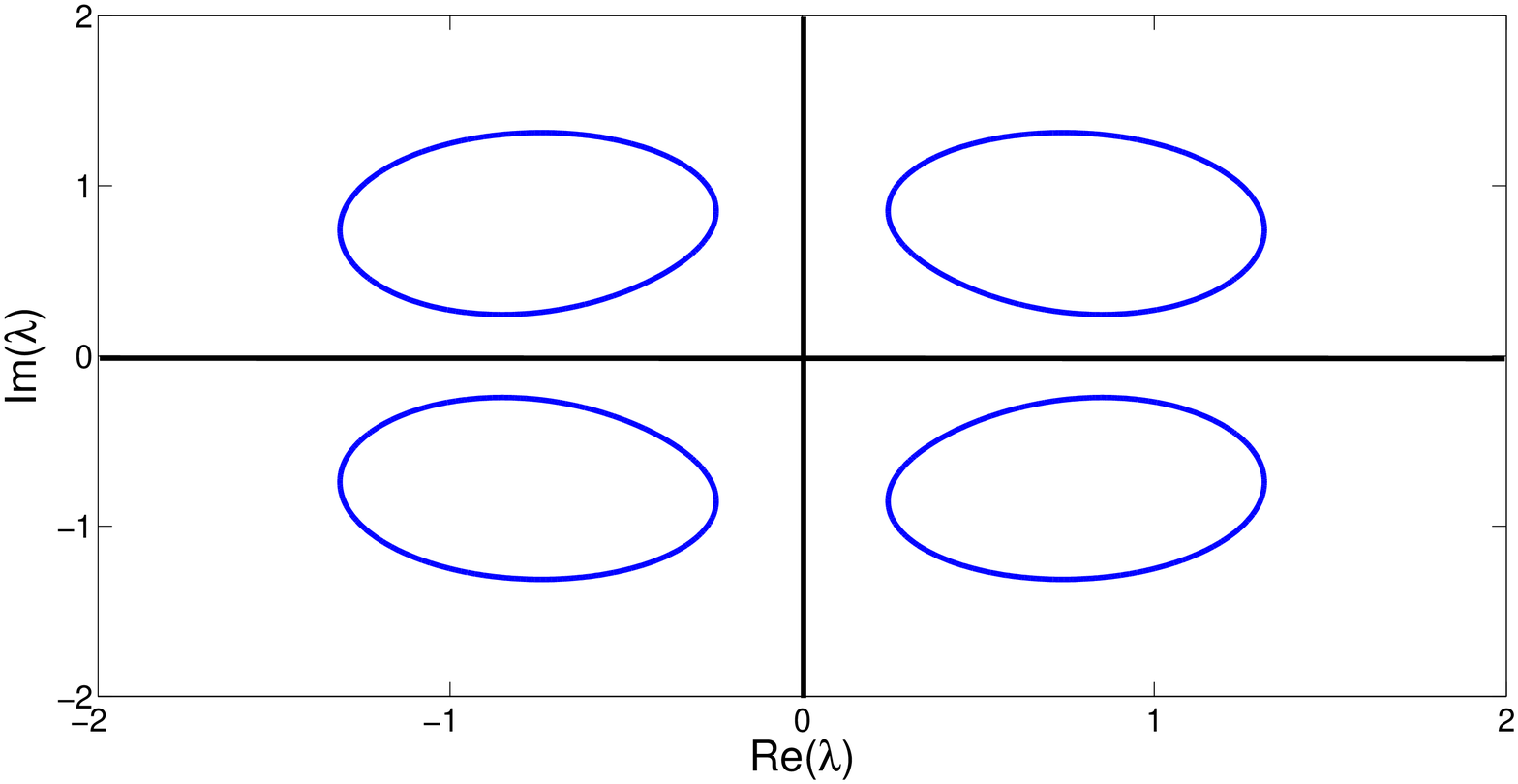}
\includegraphics[width = 3in, height = 2in]{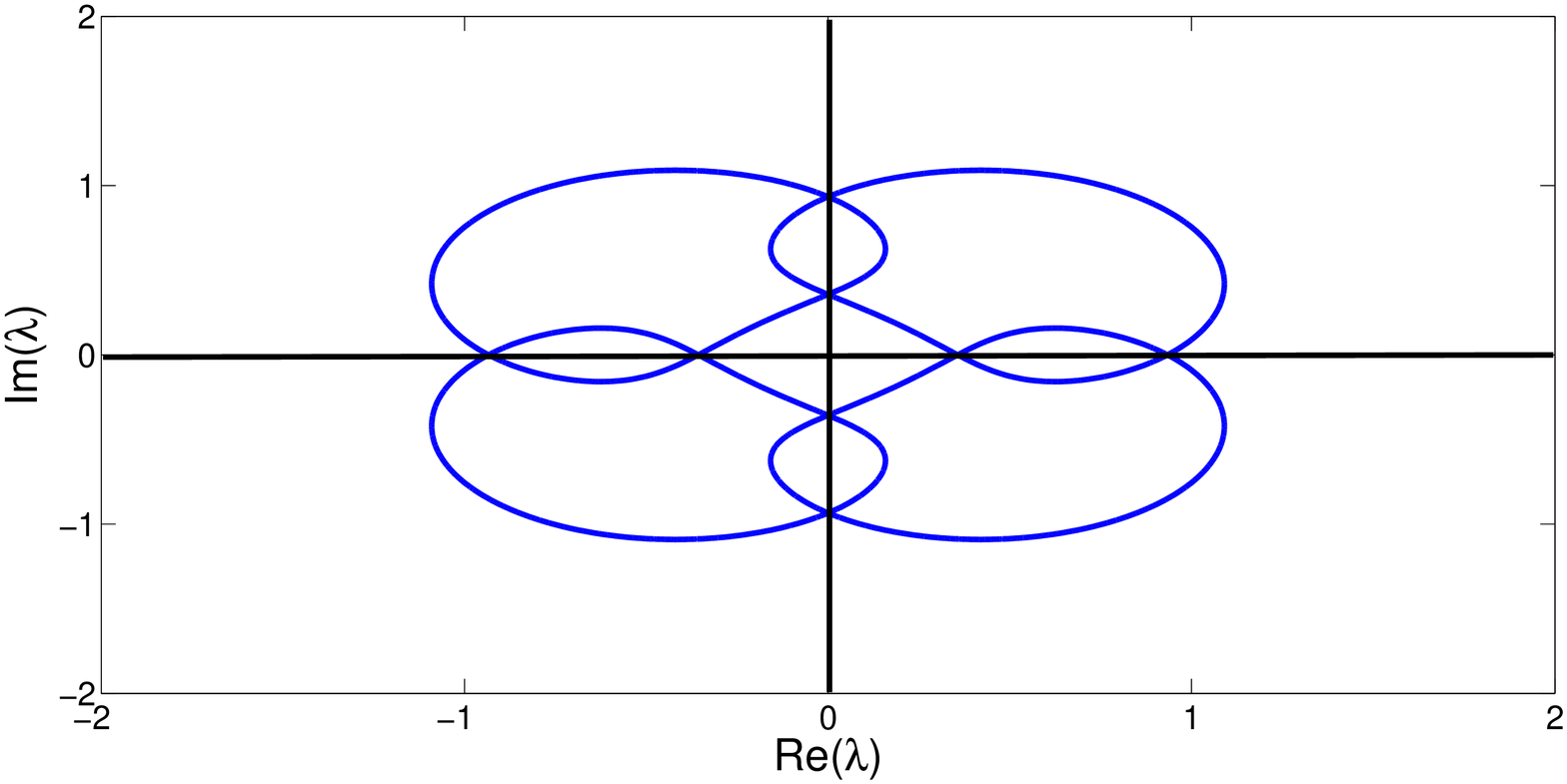}
\caption{Solutions to the transcendental equation (\ref{transc}) in
the complex plane for $h = 2$ (left) and $h = 6$ (right). Each curve encloses
a point $\lambda_0$ where either $\mu_+(\lambda_0) = 0$ or $\mu_-(\lambda_0) = 0$.} \label{fig-plane}
\end{figure}

\subsection{One-soliton solutions}

Fix $\lam_1 \in \mathbb{C}$ such that $\mu_{\pm}(\lam_1) \neq 0$ and $\lam_1^2 \notin \mathbb{R}$.
Taking a general solution for $\Phi(\lam_1) = (f,g)^T$, we write $f$ and $g$ in the form:
\begin{equation}
f_n(t) = \alpha_1 e^{\xi_{1,n}(t)}, \quad g_n(t) = \beta_1 e^{\eta_{1,n}(t)},
\label{LPsol1}
\end{equation}
where
\begin{equation}
\label{xi-eta}
\xi_{1,n}(t) = n \log \Big(\lam_1+\frac{2 \mi}{h\lam_1}\Big)+\frac{\mi}{2} \lam_1^2 t, \quad
\eta_{1,n}(t) = n \log \Big(-\lam_1+\frac{2 \mi}{h \lam_1}\Big)-\frac{\mi t}{2 \lam_1^2},
\end{equation}
and $\alpha_1,\beta_1 \in \mathbb{C}\backslash\{0\}$ are parameters. Without loss of generality,
we set $\lambda_1 = \delta_1 e^{\mi \theta_1/2}$ with some $\delta_1 > 0$ and $\theta_1 \in (0,\pi)$.
Substituting Eq.~\eref{LPsol1} into Eqs.~\eref{PoTra} yields the exact one-soliton solution
in the form:
\begin{subequations}
\begin{align}
& U^{[1]}_n = -\frac{4 \mi \delta_1 \alpha_1 \bar{\beta}_1 \sin \theta_1 e^{\mi \theta_1/2}}{
|\beta_1|^2 (2 + \mi h \delta_1^2 e^{\mi \theta_1}) e^{\eta_{1,n} - \xi_{1,n}} +
|\alpha_1|^2 (2 e^{\mi \theta_1} + \mi h \delta_1^2 ) e^{-\bar{\eta}_{1,n} + \bar{\xi}_{1,n}}},  \label{PoTr1b}  \\
& R^{[1]}_{n} = -\frac{2 \mi \delta_1 \alpha_1 \bar{\beta}_1 \sin \theta_1 e^{\mi \theta_1/2}}{
|\beta_1|^2 e^{\eta_{1,n} - \xi_{1,n}} + |\alpha_1|^2 e^{-\bar{\eta}_{1,n} + \bar{\xi}_{1,n} + \mi \theta_1}}, \label{PoTr2b}\\
& Q^{[1]}_{n} = -\frac{2 \mi \alpha_1 \bar{\beta}_1 \sin \theta_1 e^{\mi \theta_1/2}}{\delta_1(
|\beta_1|^2 e^{\eta_{1,n} - \xi_{1,n} + \mi \theta_1} +
|\alpha_1|^2 e^{-\bar{\eta}_{1,n} + \bar{\xi}_{1,n}})}, \label{PoTr3b}
\end{align}
\label{PoTrb}
\end{subequations}
where
\begin{equation*}
\xi_{1,n}(t) - \eta_{1,n}(t) = n \log \Big( \frac{2 - \mi h \delta_1^2 e^{\mi \theta_1}}{2 + \mi h \delta_1^2 e^{\mi \theta_1}} \Big)
+ \frac{\mi}{2} \big( \delta_1^2 + \frac{1}{\delta_1^2} \big) \cos \theta_1 t
- \frac{1}{2} \big( \delta_1^2 - \frac{1}{\delta_1^2} \big) \sin \theta_1 t.
\end{equation*}
Fig.~\ref{Fig1a}--\ref{Fig1c} presents the one-soliton solutions (\ref{PoTrb}) for
$\alpha_1 = 1$, $\beta_1 = 1 + \mi$, $\lambda_1 = 2e^{\frac{\pi}{5}\mi}$, and $h = 1$.

\begin{figure}[h!]
 \centering
\subfigure[]{\label{Fig1a}
\includegraphics[width=2in]{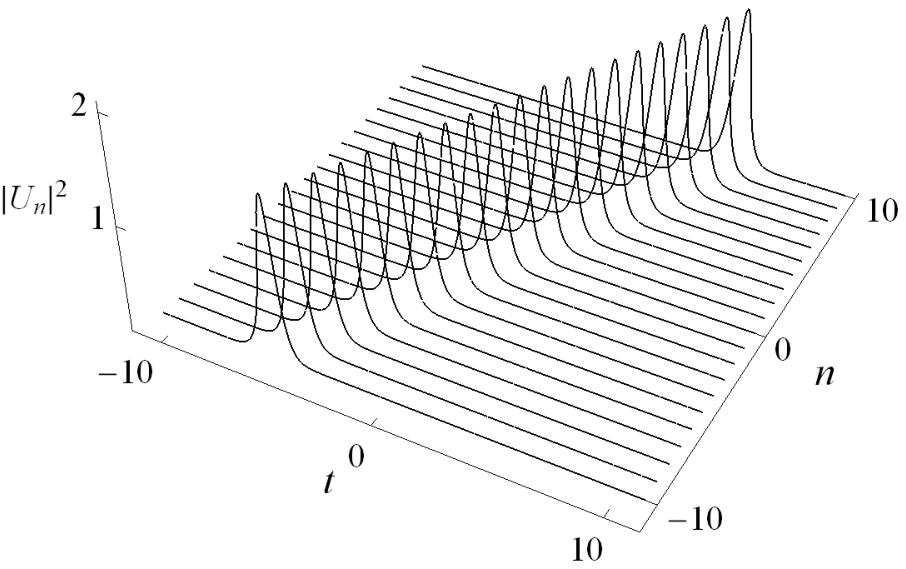}}\hfill
\subfigure[]{ \label{Fig1b}
\includegraphics[width=2in]{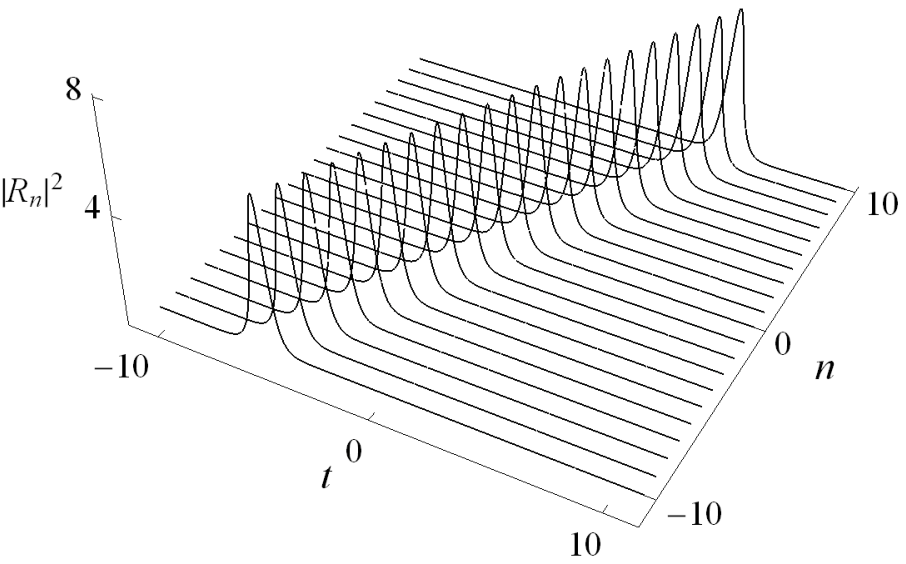}}\hfill
\subfigure[]{ \label{Fig1c}
 \includegraphics[width=2in]{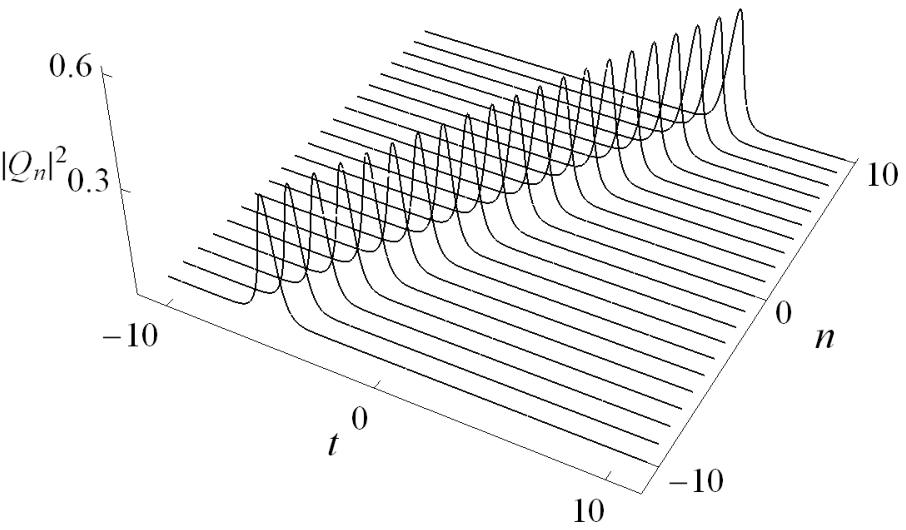}}
\caption{An example of the one-soliton solutions~\eref{PoTrb}. The following
components are shown: $|U_n|^2$ (left), $|R_n|^2$ (middle), and $|Q_n|^2$ (right).}
\end{figure}

Let us check that the discrete solitons (\ref{PoTrb}) recover solitons
of the continuous MTM system (\ref{MTM-continuous}). In order to simplify the computations, we set
$\delta_1 = 1$, which corresponds to the case of stationary solitons \cite{Yusuke1,Yusuke2}.
By defining $x_n = hn$, $n \in \mathbb{Z}$ and taking the limit $h \to 0$, we obtain
for $\delta_1 = 1$:
\begin{subequations}
\begin{align}
& U^{[1]}_n \to U(x,t) = -\frac{2 \mi \alpha_1 \bar{\beta}_1 \sin \theta_1 e^{\mi \cos \theta_1 (t-x)}}{
|\alpha_1|^2 e^{\sin \theta_1 x + \mi \theta_1/2} + |\beta_1|^2 e^{-\sin \theta_1 x - \mi \theta_1/2}},  \label{PoTr1c}  \\
& R^{[1]}_{n} \to R(x,t) = -\frac{2 \mi \alpha_1 \bar{\beta}_1 \sin \theta_1 e^{\mi \cos \theta_1 (t-x)}}{
|\alpha_1|^2 e^{\sin \theta_1 x + \mi \theta_1/2} + |\beta_1|^2 e^{-\sin \theta_1 x - \mi \theta_1/2}}, \label{PoTr2c}\\
& Q^{[1]}_{n} \to Q(x,t) = -\frac{2 \mi \alpha_1 \bar{\beta}_1 \sin \theta_1 e^{\mi \cos \theta_1 (t-x)}}{
|\alpha_1|^2 e^{\sin \theta_1 x - \mi \theta_1/2} + |\beta_1|^2 e^{-\sin \theta_1 x + \mi \theta_1/2}}, \label{PoTr3c}
\end{align}
\label{PoTrc}
\end{subequations}
which agree with the MTM solitons in the continuous system (\ref{MTM-continuous}).
Parameters $\alpha_1, \beta_1 \in \mathbb{C}\backslash\{0\}$ determine translations in space and rotation in time,
whereas $\theta_1 \in (0,\pi)$ determines
the frequency $\omega_1 := \cos \theta_1 \in (-1,1)$ of the continuous MTM solitons.
In the limit $\omega_1 \to 1$ ($\theta_1 \to 0$), the MTM soliton (\ref{PoTrc}) degenerates to
the zero solution, whereas in the limit $\omega_1 \to -1$ ($\theta_1 \to \pi$) and $\alpha_1=\beta_1$,
it becomes the algebraic solitons:
\begin{equation}
U(x,t) \to U_a(x,t) = -\frac{e^{-\mi (t-x)}}{x - \mi/2}.
\label{AlgSoliton}
\end{equation}

Discrete solitons (\ref{PoTrb}) enjoy the same properties as the continuous solitons.
In particular, let us recover the discrete algebraic soliton for the case $\alpha_1 = \beta_1$ and
$\delta_1 = 1$ in the limit $\theta_1 \to \pi$. By setting $\theta_1 = \pi - \epsilon$
and expanding to the first order in $\epsilon$, we obtain from (\ref{PoTr1b})
$$
U^{[1]}_n = \frac{4 ( \epsilon + \mathcal{O}(\epsilon^2))  e^{- \mi t}}{
(2 - \mi h - \epsilon h + \mathcal{O}(\epsilon^2)) \left(\frac{2 - \mi h + \mi \epsilon (2 + \mi h)/2
+ \mathcal{O}(\epsilon^2)}{2+ \mi h + \mi \epsilon (2 - \mi h)/2 + \mathcal{O}(\epsilon^2)} \right)^n
- (2 - \mi h - 2 \mi \epsilon + \mathcal{O}(\epsilon^2)) \left(\frac{2 - \mi h - \mi \epsilon (2 + \mi h)/2
+ \mathcal{O}(\epsilon^2)}{2+ \mi h - \mi \epsilon (2 - \mi h)/2 + \mathcal{O}(\epsilon^2)} \right)^n}.
$$
This expression yields in the limit $\epsilon \to 0$ the discrete algebraic soliton
\begin{equation}
U^{[1]}_n \to [U_a]_n= -\frac{4 e^{- \mi t}}{\frac{8 n h (2 - \mi h)}{4 + h^2} - 2 \mi + h}
\left(\frac{2 + \mi h}{2 - \mi h} \right)^n.
\label{AlgSoliton-Discrete}
\end{equation}
If $x_n = hn$, $n \in \mathbb{Z}$,
the discrete algebraic soliton (\ref{AlgSoliton-Discrete}) reduces in the limit $h \to 0$ to the continuous algebraic soliton
(\ref{AlgSoliton}). Similarly, one can prove that the discrete soliton (\ref{PoTrb}) 
degenerates to the zero solution in the limit $\theta_1 \to 0$.

\subsection{Two-soliton solutions}

In order to construct the two-soliton solutions, one needs to use the one-fold Darboux transformation
(\ref{PoTr}) twice. Fix $\lam_1, \lam_2 \in \mathbb{C}\backslash\{0\}$ such that $\mu_{\pm}(\lam_{1,2}) \neq 0$,
$\lam_{1,2}^2 \notin \mathbb{R}$, $\lam_2 \neq \pm \lam_1$, and $\lam_2 \neq \pm \bar{\lam}_1$.
A general solution of the Lax pair~\eref{LP1} with $\lam = \lam_1$ and $\lam=\lam_2$ at zero background is written
in the form
\begin{equation}
[\Phi(\lam_1)]_n(t) = \left( \begin{array}{c} \alpha_1 e^{\xi_{1,n}(t)} \\
 \beta_1 e^{\eta_{1,n}(t)} \end{array} \right), \quad
[\Phi(\lam_2)]_n(t) = \left( \begin{array}{c} \alpha_2 e^{\xi_{2,n}(t)} \\
 \beta_2 e^{\eta_{2,n}(t)} \end{array} \right), \label{LPsol2}
\end{equation}
where $\xi_{j,n}$ and $\eta_{j,n}$ with $j = 1,2$ are given by (\ref{xi-eta}) for $\lambda_{1,2}$,
and $\alpha_{1,2},\beta_{1,2} \in \mathbb{C}\backslash\{0\}$ are parameters.

By using the one-fold Darboux transformation~\eref{PoTr} with zero potentials, $\lambda = \lambda_1$,
and $\Phi = \Phi(\lam_1)$, we obtain the one-soliton solutions $(U^{[1]}_n, R^{[1]}_n, Q^{[1]}_n)$ in the form (\ref{PoTrb}).
The transformed eigenfunction $\Phi^{[1]}(\lam_2) = T^{[1]}(\lam_2) \Phi(\lam_2)$ satisfies the
Lax pair~\eref{LP1} with the potentials $(U^{[1]}_n, R^{[1]}_n, Q^{[1]}_n)$  and $\lam=\lam_2$.
By using the one-fold Darboux transformation~\eref{PoTr} with
$(U_n, R_n, Q_n)$ replaced by $(U^{[1]}_n, R^{[1]}_n, Q^{[1]}_n)$,
$\lam_1$ replaced by $\lam_2$, and $\Phi(\lam_1)$ replaced by
$\Phi^{[1]}(\lam_2)$,
we obtain the two-soliton solutions $(U^{[2]}_n\,, R^{[2]}_n\,,Q^{[2]}_n)$ in the explicit form (which is not given here).

Fig.~\ref{Fig2a}--\ref{Fig2c} shows the two-soliton solutions for $\alpha_1 = 1$, $\beta_1 = 1 + \mi$, $\alpha_2 = 1$, $\beta_2 = 1$,
$\lambda_1 = \sqrt{3} e^{\mi \pi/6}$, $\lambda_2 = \sqrt{5} e^{ \frac{\mi\arctan 2}{2}}$, and $h = 1$.
The two-soliton solutions feature elastic collisions of two individual solitons with preservation of their shapes. 
Such collisions are very common in integrable equations including the continuous MTM system (\ref{MTM}).

\begin{figure}[h!]
 \centering
\subfigure[]{\label{Fig2a}
\includegraphics[width=2in]{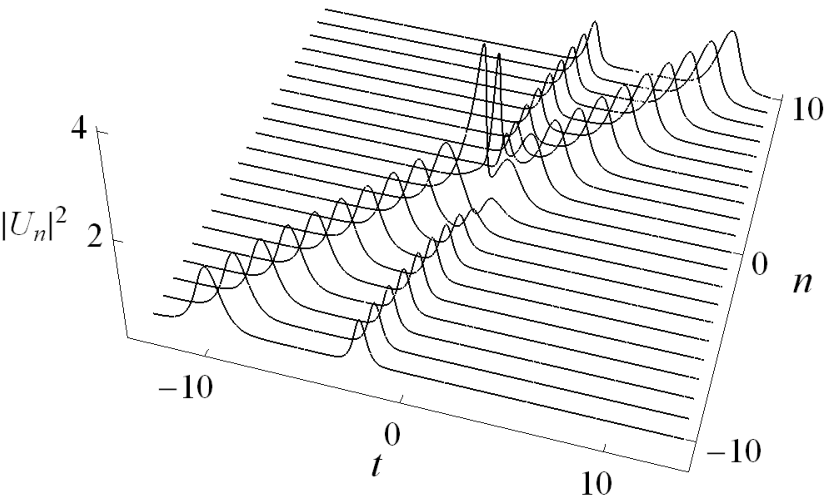}}\hfill
\subfigure[]{ \label{Fig2b}
\includegraphics[width=2in]{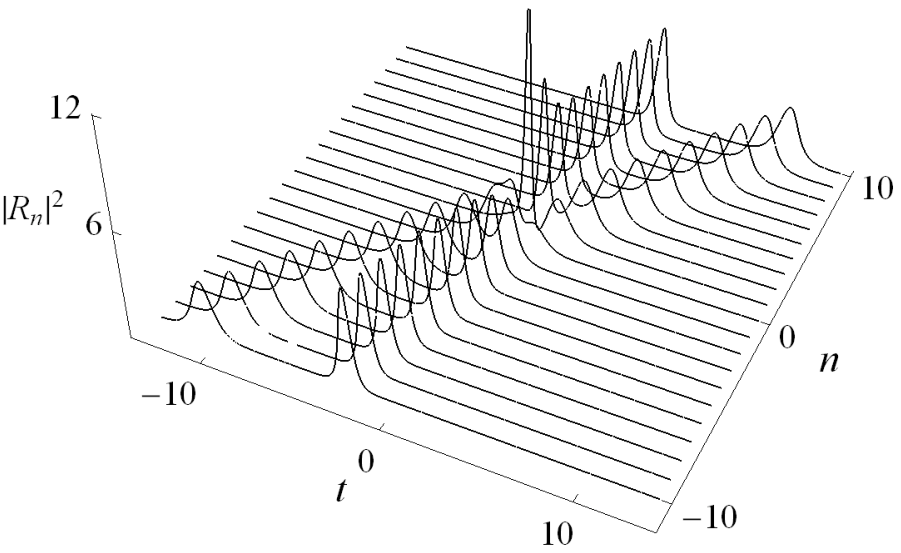}}\hfill
\subfigure[]{ \label{Fig2c}
 \includegraphics[width=2in]{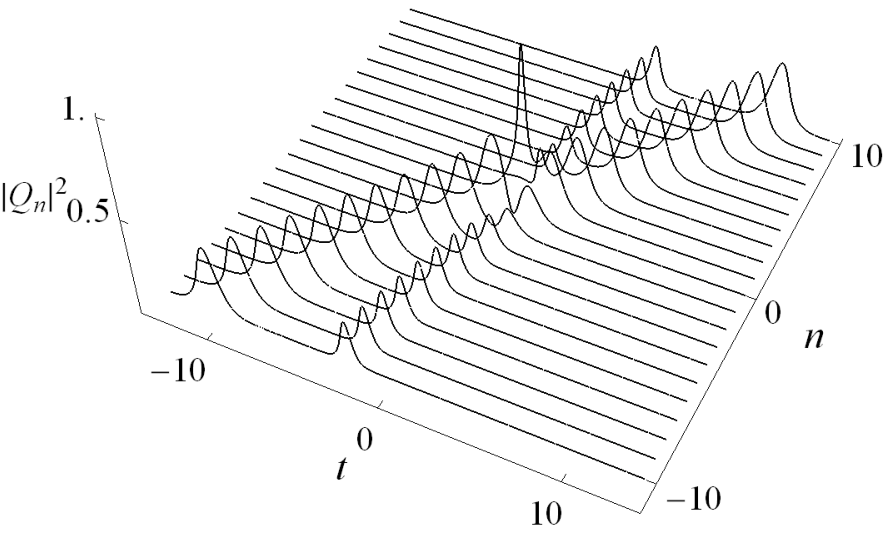}}
\caption{An example of the two-soliton solutions.}
\end{figure}

\section{Soliton solutions on nonzero constant background}

Here we use the one-fold Darboux transformation of Theorem \ref{theorem-main} and
construct soliton solutions on nonzero constant background
$(U,R,Q) = (\rho,\rho,\rho^{-1})$, where $\rho > 0$ is a real parameter.
Similarly to Section 3, we prove that the nonzero constant background is linearly stable
in the semi-discrete MTM system (\ref{MTM-discrete}) for every $\rho > 0$,
construct Jost solutions of the Lax pair (\ref{LP1}) at nonzero constant background,
and then finally obtain the exact expressions for one-soliton solutions.

\subsection{Stability of nonzero constant background}

Linearization of the semi-discrete MTM system (\ref{MTM-discrete}) at the nonzero constant
background $(U,R,Q) = (\rho,\rho,\rho^{-1})$ with $\rho > 0$ yields the linear system of equations:
\begin{equation}
\label{MTM-linear-nonzero}
\left\{ \begin{array}{l}
\displaystyle
4 \mi \frac{d u_n}{dt} + 2 \left(\rho^2 - \frac{1}{\rho^2} \right) u_n
+ \left( 1 + \frac{\mi h \rho^2}{2} \right) \left( \frac{2 \mi}{h} r_{n+1} - \bar{q}_{n+1} \right)
- \left( 1 - \frac{\mi h \rho^2}{2} \right) \left( \frac{2 \mi}{h} r_n + \bar{q}_n \right) = 0, \\
\displaystyle
 \left( 1 + \frac{\mi h \rho^2}{2} \right) \bar{q}_{n+1} -  \left( 1 - \frac{\mi h \rho^2}{2} \right) \bar{q}_n
 + \mi h u_n = 0, \\
\displaystyle
 \left( 1 + \frac{\mi h \rho^2}{2} \right) r_{n+1} +  \left( 1 - \frac{\mi h \rho^2}{2} \right) r_n - 2 u_n = 0. \end{array} \right.
\end{equation}
By using the discrete Fourier transform on the lattice (\ref{Fourier}), we close the linear system (\ref{MTM-linear-nonzero}) at
the following differential equation with parameter $\theta \in (-\pi,\pi)$:
\begin{equation}
\label{closed-eq}
\mi h \frac{d \hat{u}}{d t} + \frac{h}{2} \left(\rho^2 - \frac{1}{\rho^2} \right) \hat{u}
+ \left( \frac{h^2}{4} \frac{\cos \frac{\theta}{2} - \frac{h \rho^2}{2} \sin \frac{\theta}{2}}{\sin \frac{\theta}{2} + \frac{h \rho^2}{2} \cos \frac{\theta}{2}}
- \frac{\sin \frac{\theta}{2} + \frac{h \rho^2}{2} \cos \frac{\theta}{2}}{\cos \frac{\theta}{2} -
 \frac{h \rho^2}{2} \sin \frac{\theta}{2}} \right) \hat{u} = 0.
\end{equation}
The dispersion relation following from linear equation (\ref{closed-eq}) is purely real, which
implies that the nonzero constant background is linearly stable for every $\rho > 0$.
Note that the linear equation (\ref{closed-eq}) does not reduce to equation (\ref{closed-eq-zero})
in the limit $\rho \to 0$ because the nonzero constant background $(U,R,Q) = (\rho,\rho,\rho^{-1})$
is singular in this limit, hence the variable $q$ in the linearized system (\ref{MTM-linear})
is replaced by $\bar{q}$ in the system (\ref{MTM-linear-nonzero}).

Note that $(u,v) = (\rho,\rho^{-1})$ is also the nonzero constant solution of the continuous MTM system (\ref{MTM}). 
However, computations similar to those in (\ref{MTM-linear-nonzero}) and (\ref{closed-eq}) 
show that the nonzero constant background for any $\rho > 0$ is modulationally unstable. This is different 
from the conclusion on the nonzero constant background in the semi-discrete MTM system (\ref{MTM-discrete}). 

\subsection{Solutions of the Lax pair (\ref{LP1}) at nonzero constant background}

Solving Lax pair~\eref{LP1} with the potentials $(U,R,Q) = (\rho,\rho,\rho^{-1})$, we have two linearly independent solutions:
\begin{equation}
\label{Jost-nonzero}
[\Phi_+(\lam)]_n(t) = \alpha \left( \begin{array}{c} \rho \\
-\lambda \end{array} \right) \mu_+^n e^{\frac{\mi}{2} \left( \frac{1}{\rho^2} - \rho^2\right) t}, \quad
[\Phi_-(\lam)]_n(t) = \beta \left( \begin{array}{c} \lambda \\
\rho \end{array} \right) \mu_-^n e^{\frac{\mi}{2} \left( \lam^2 - \frac{1}{\lam^2} \right) t},
\end{equation}
where $\alpha,\beta \in \mathbb{C}\backslash\{0\}$ are parameters and
$$
\mu_+(\lam) := \left( \frac{2\,\mi}{h \lam} - \lam \right) \frac{1 + \frac{\mi h \rho^2}{2}}{1 - \frac{\mi h \rho^2}{2}}, \quad
\mu_-(\lam) := \frac{2\,\mi}{h \lam} + \lam.
$$
Similarly to the case of zero potentials, we say that $\Phi(\lam)$ is a Jost function
if $\lambda \in \mathbb{C}$ yields either $|\mu_+(\lam)| = 1$ or $|\mu_-(\lam)| = 1$.
Interestingly, the constraints $|\mu_{\pm}(\lam)| = 1$ with $\lambda = |\lambda| e^{\mi \theta/2}$ yield
the same equation (\ref{transc}).
Hence, any point on each curve of the Lax spectrum shown on Fig. \ref{fig-plane} gives
one Jost function in (\ref{Jost-nonzero}) which remains bounded in the limit $|n| \to \infty$.
The function of $\Phi_+(\lam)$ is always bounded in the limit $|t| \to \infty$ since $\rho > 0$.
On the other hand, $\Phi_-(\lam)$ is bounded as $|t| \to \infty$  if and only if $\lambda^2 \in \mathbb{R}$,
and no such Jost functions exist for $\Phi_-(\lam)$ if $h<4$.

\subsection{One-breather solutions}

Fix $\lam_1 \in \mathbb{C}$ such that $\mu_{\pm}(\lam_1) \neq 0$ and $\lam_1^2 \notin \mathbb{R}$.
Let $\Phi(\lam_1) = (f,g)^T$ be the general solution of Lax pair~\eref{LP1} with $(U,R,Q) = (\rho,\rho,\rho^{-1})$ and $\lam=\lam_1$.
We write $f$ and $g$ in the form
\begin{equation}
f_{1,n} = \alpha_1 \rho\, e^{\mu_{1,n}(t)} +  \beta_1 \lambda_1 e^{\nu_{1,n} (t)}, \quad
g_{1,n} = - \alpha_1 \lambda_1  e^{\mu_{1,n}(t)}  + \beta_1\rho\, e^{\nu_{1,n}(t)},
\label{LPsol3}
\end{equation}
with
\begin{align*}
& \mu_{1,n}(t)= n \log\left[ \left( \frac{2\,\mi}{h \lam_1} - \lam_1 \right) \frac{1 + \frac{\mi h \rho^2}{2}}{1 - \frac{\mi h \rho^2}{2}}\right] +\frac{\mi}{2}\left(\frac{1}{\rho ^2}-\rho ^2\right) t, \\[1.5mm]
& \nu_{1,n}(t)= n \log \left(\lambda_1+\frac{2 \mi}{h \lambda_1}\right)+\frac{\mi}{2}
   \left(\lambda_1^2-\frac{1}{\lambda_1^2}\right) t,
\end{align*}
where $\alpha_1,\beta_1 \in \mathbb{C}\backslash\{0\}$ are parameters.
Substituting Eq.~\eref{LPsol3} into Eqs.~\eref{PoTr}, we obtain the one-breather solutions at nonzero constant background as follows:
\begin{subequations}
\begin{align}
& U^{[1]}_n =-\frac{|\alpha_1|^2\rho \bar{\lambda}_1 h_{\lam_1} \bar{\chi}_1  e^{\Theta_{1,n}} +
|\beta_1|^2 \rho \lambda_1 h_{\bar{\lam}_1} \chi_1 e^{-\Theta_{1,n}}
+ \bar{\alpha}_1\beta_1|\lambda_1|^2 h_{\rho} (\lambda_1^2-\bar{\lambda}_1^2)
e^{-\mi\,\Xi_{1,n}} }{|\alpha_1|^2 \lambda_1 h_{\lam_1} \bar{\chi}_1
e^{\Theta_{1,n}} + |\beta_1|^2 \bar{\lambda}_1 h_{\bar{\lam}_1}  \chi_1
e^{-\Theta_{1,n}}- \bar{\alpha}_1 \beta_1\rho h_{\rho} (\lambda_1^2-\bar{\lambda}_1^2)
e^{-\mi\,\Xi_{1,n}} }, \label{breathera} \\
& R^{[1]}_n = -\frac{|\alpha_1|^2 \rho \bar{\lam}_1 \bar{\chi}_1 e^{\Theta_{1,n}} +
|\beta_1|^2 \rho \lam_1  \chi_1 e^{-\Theta_{1,n}} - \bar{\alpha}_1\beta_1 |\lam_1|^2(\lambda_1^2-\bar{\lambda}_1^2) e^{-\mi\,\Xi_{1,n}}}{|\alpha_1|^2 \lam_1 \bar{\chi}_1 e^{\Theta_{1,n}} +
|\beta_1|^2 \bar{\lam}_1  \chi_1 e^{-\Theta_{1,n}}+ \bar{\alpha}_1\beta_1 \rho (\lambda_1^2-\bar{\lambda}_1^2) e^{-\mi\,\Xi_{1,n}}},
 \label{breatherb} \\
& Q^{[1]}_n = -\frac{|\alpha_1|^2 \bar{\lambda}_1^3 \chi_1 e^{\Theta_{1,n}} +
|\beta_1|^2\lambda_1^3 \bar{\chi}_1 e^{-\Theta_{1,n}} +  \alpha_1 \bar{\beta}_1 \rho ^3  (\lambda_1^2-\bar{\lambda}_1^2)
e^{\mi\,\Xi_{1,n}}}
{ \rho |\lambda_1|^2 \left[
|\alpha_1|^2 \bar{\lambda}_1 \chi_1 e^{\Theta_{1,n}} +
|\beta_1|^2 \lambda_1 \bar{\chi}_1 e^{-\Theta_{1,n}} -
\alpha_1 \bar{\beta}_1  \rho (\lambda_1^2-\bar{\lambda}_1^2)
e^{\mi\,\Xi_{1,n}} \right]}, \label{breatherc}
\end{align} \label{breather}
\end{subequations}
where
\begin{align*}
& \Theta_{1,n}=\Re(\mu_{1,n}-\nu_{1,n}), \quad \Xi_{1,n}=\Im(\mu_{1,n}-\nu_{1,n}), \\
& \chi_1=\rho ^2 + \lambda_1^2,  \quad \bar{\chi}_1=\rho ^2 + \bar{\lambda}_1^2, \\
& h_{\lam_1}=-2 \mi + h\lambda_1^2, \quad h_{\bar{\lam}_1}= -2 \mi + h\bar{\lambda}_1^2, \quad h_{\rho}= 2 \mi + h \rho ^2.
\end{align*}
Due to the presence of the oscillatory terms $e^{\mi\,\Xi_{1,n}}$ and $e^{-\mi\,\Xi_{1,n}}$,
solutions~\eref{breather}, in general, exhibit the localized breathers which oscillate periodically both in $n$ and $t$.
Fig.~\ref{Fig3a}--\ref{Fig3c} illustrates the one-breather solutions (\ref{breather}) at the constant background for
$\alpha_1 = 1$, $\beta_1 = 1 + \mi$, $\rho=1$, $\lambda_1 = 2 e^{\mi \pi/8}$,  and $h = 3/4$.

\begin{figure}[h!]
 \centering
\subfigure[]{\label{Fig3a}
\includegraphics[width=2in]{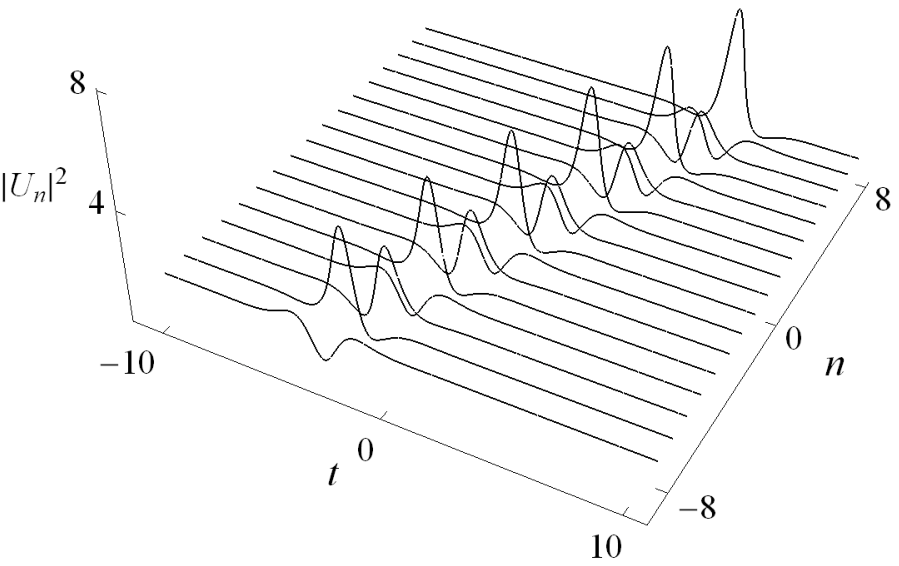}}\hfill
\subfigure[]{ \label{Fig3b}
\includegraphics[width=2in]{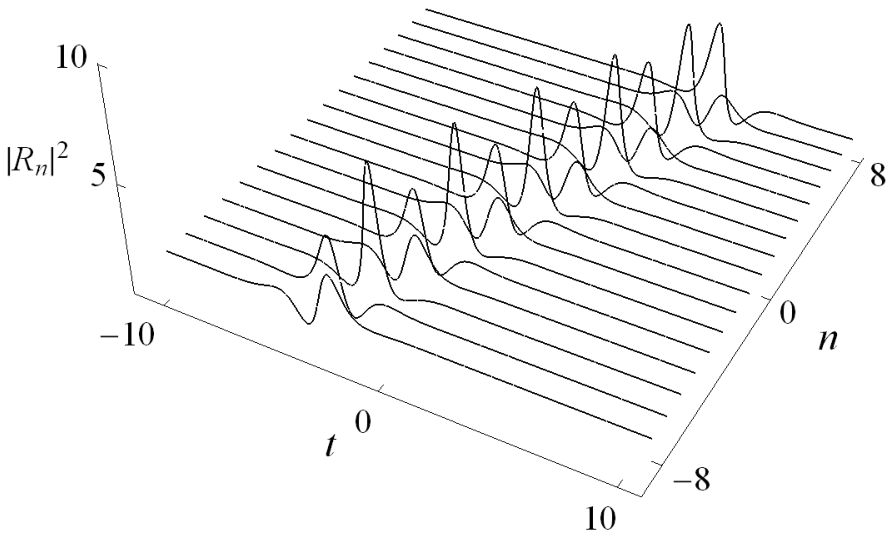}}\hfill
\subfigure[]{ \label{Fig3c}
 \includegraphics[width=2in]{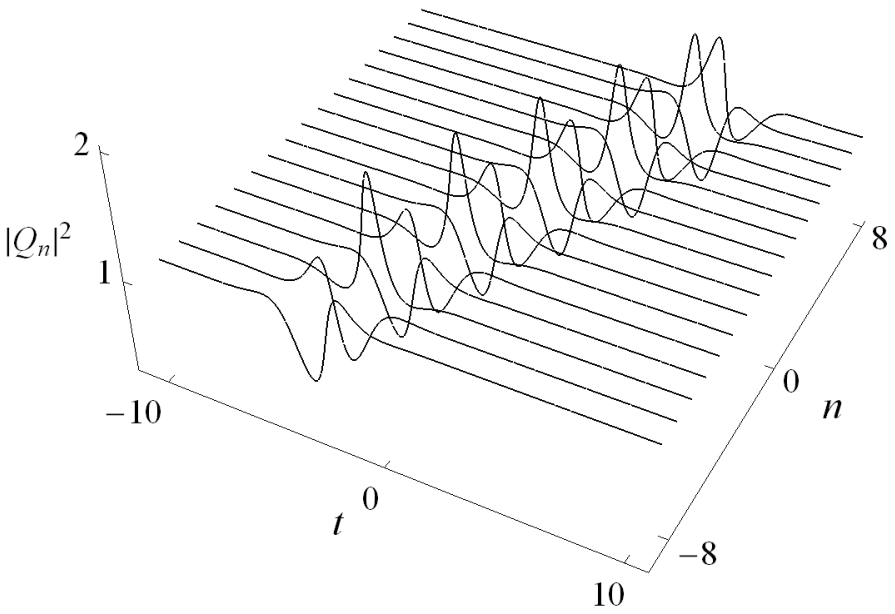}}
\caption{An example of the one-breather solutions~\eref{breather} at the nonzero background. }
\end{figure}

No periodic oscillations occur in the one-breather solutions~\eref{breather} if and only if $\Xi_{1,n}=0$.
In this case, solutions~\eref{breather} describe one-solitons illustrated on Fig.~\ref{Fig5a}--\ref{Fig5c} for $\alpha_1 = 1$, $\beta_1 = 1 + \mi$,
$\rho=\frac{2^{3/4}}{\sqrt[4]{7}}$, $\lambda_1 = \sqrt[4]{\frac{7}{6}} e^{\frac{5 \pi }{12}\mi}$,  and $h = \sqrt{3}$.

\begin{figure}[h!]
 \centering
\subfigure[]{\label{Fig5a}
\includegraphics[width=2in]{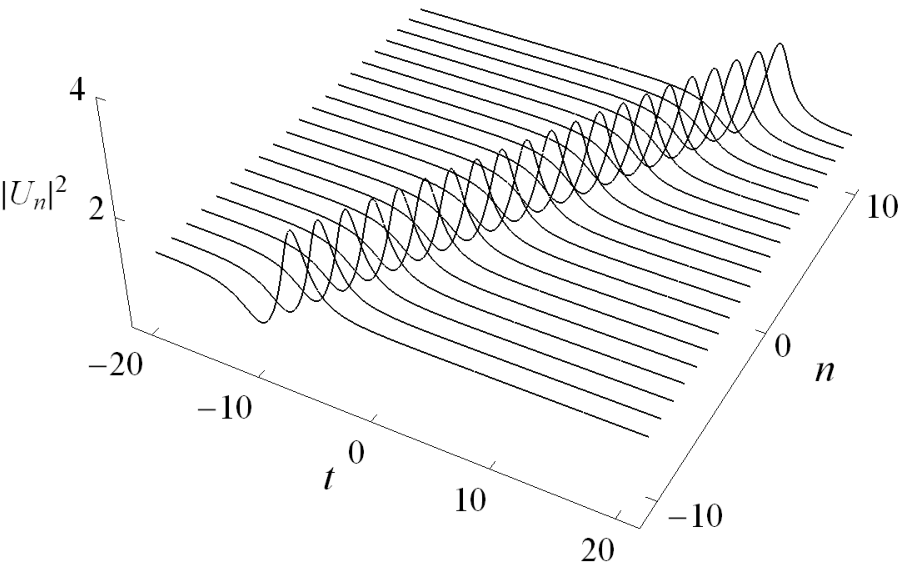}}\hfill
\subfigure[]{ \label{Fig5b}
\includegraphics[width=2in]{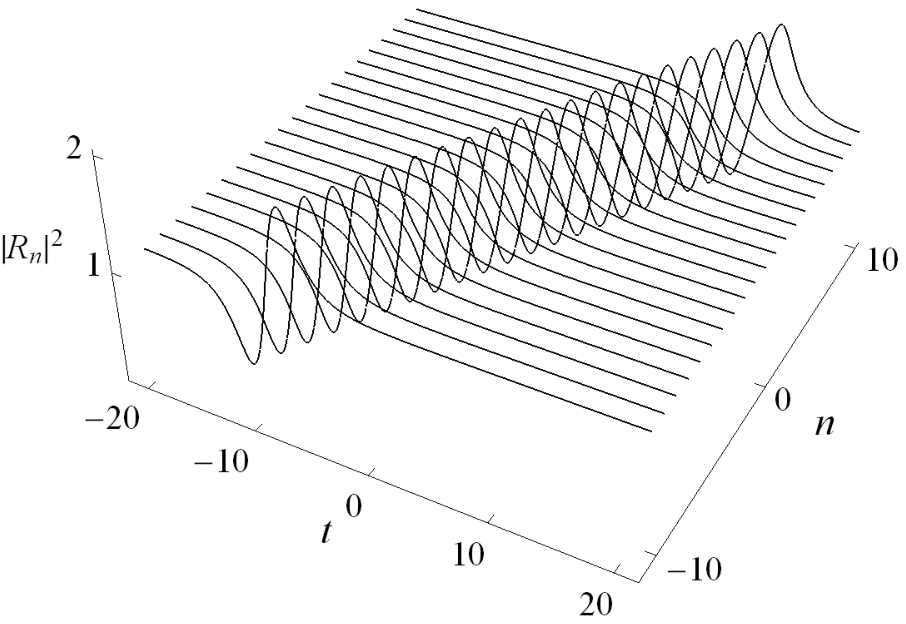}}\hfill
\subfigure[]{ \label{Fig5c}
 \includegraphics[width=2in]{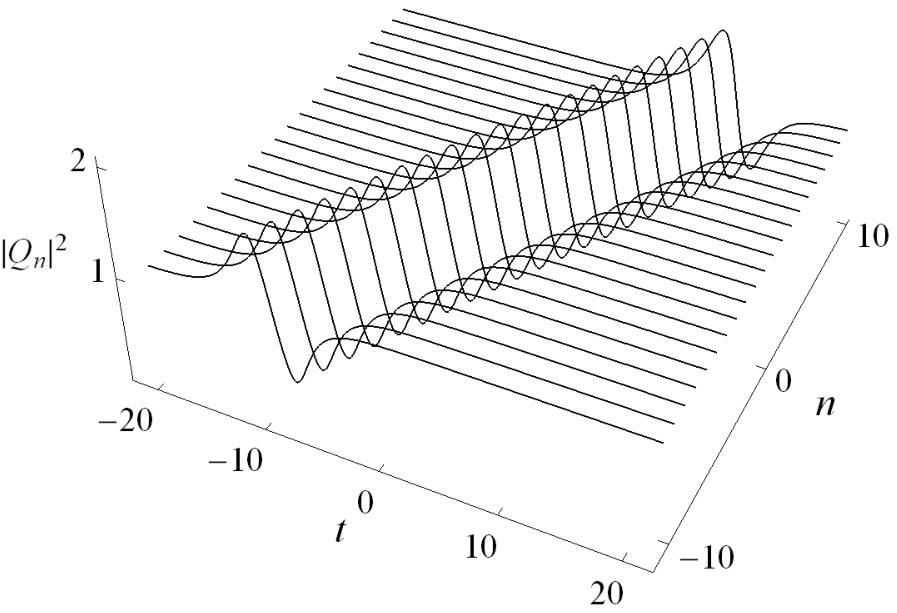}}
\caption{An example of the one-breather solutions~\eref{breather} without periodic oscillations.}
\end{figure}

We show that the one-breather solutions~\eref{breather} feature no periodic oscillations
if the modulus and argument of $\lam_1$ are given by
\begin{align}
|\lam_1| = \frac{1}{\rho}\sqrt{\frac{2}{h}}, \quad
\arg(\lam_1) = \frac{1}{2} \arccos\left(2h \frac{1-\rho ^4}{4- h^2 \rho ^4}\right)
\label{lam1}
\end{align}
in the two regions described by
\begin{align}
\mbox{\rm either } \; h > \frac{2}{\rho^4}, \;\; \rho < 1,\,\, \text{or}\,\, 0 < h < \frac{2}{\rho^4}, \;\; \rho > 1.
\label{lam2}
\end{align}
Note that the two regions intersect at $\rho = 1$, $h = 2$, for which $|\lam_1| = 1$
whereas $\arg(\lam_1)$ is not determined. In fact, we show that
$\arg(\lam_1) \in (\frac{\pi}{4}, \frac{\pi}{2})$.
The existence region for non-oscillating one-soliton solutions (\ref{breather})
on the $(h, \rho)$ plane is displayed in Fig.~\ref{Fig4}.

\begin{figure}[h!]
 \centering
\includegraphics[width = 3in]{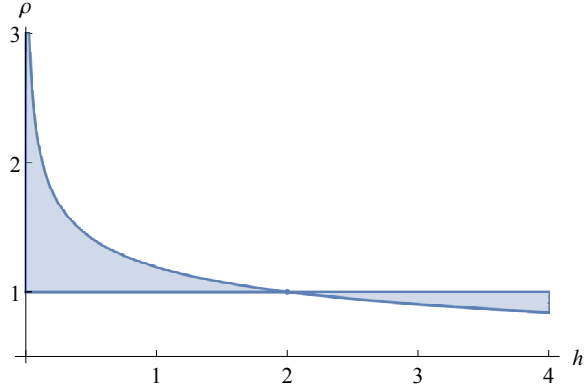}
\caption{Region on the $(h, \rho)$ plane given by ~\eref{lam2}. }
\label{Fig4}
\end{figure}

In order to verify (\ref{lam1}), we note that the condition $\Xi_{1,n}=0$
is equivalent to the system of two equations
\begin{eqnarray}
 \label{Cond1a}
 \left\{ \begin{array}{l}
\frac{2}{\rho ^2} -2 \rho ^2 -\bar{\lambda}_1^2-\lambda_1^2 + \frac{1}{\bar{\lambda}_1^2} + \frac{1}{\lambda_1^2}=0,\\
\frac{4\rho^2}{h |\lam_1|^2} - h |\lam_1|^2 \rho^2 + \frac{2}{h}\left(\frac{\bar\lam_1}{\lam_1}+ \frac{\lam_1}{\bar\lam_1} \right)\left(1-\frac{h^2\rho^4}{4}\right) =0,
\end{array} \right.
\end{eqnarray}
subject to the constraint
\begin{equation}
 \label{Cond1b}
\left(\frac{4}{h^2|\lam_1|^2} -|\lam_1|^2\right)\left(1-\frac{h^2\rho^4}{4}\right) - 2\rho^2 \left(\frac{\bar\lam_1}{\lam_1}+ \frac{\lam_1}{\bar\lam_1} \right) > 0.
\end{equation}
By using the polar form $\lambda_1 = \delta_1 e^{\mi \theta_1/2}$ with $\delta_1 > 0$ and $\theta_1 \in (0,\pi)$,
we rewrite the constraints (\ref{Cond1a})--(\ref{Cond1b}) in the form:
\begin{eqnarray}
 \label{Cond2a}
 \left\{ \begin{array}{l}
\frac{1}{\rho ^2} - \rho ^2 + \left(\frac{1}{\delta _1^2} -\delta _1^2 \right) \cos \theta_1 =0, \\
\delta _1^4 h^2 \rho ^2+\delta _1^2 \left(h^2 \rho ^4-4\right) \cos \theta_1-4 \rho ^2=0,
\end{array} \right.
\end{eqnarray}
subject to the constraint
\begin{equation}
 \label{Cond2b}
\frac{\left(\delta _1^4 h^2-4\right) \left(h^2 \rho ^4-4\right)}{4 \delta _1^2 h^2}-4 \rho ^2 \cos\theta_1  >0.
\end{equation}
Let us first assume that $\delta_1 \neq 1$, in which case the first equation in (\ref{Cond2a}) gives a unique solution
for $\theta_1$:
\begin{equation}
 \label{Cond3a}
\cos \theta_1 = \frac{\rho^2-\rho^{-2}}{\delta_1^{-2} - \delta_1^2}.
\end{equation}
Substituting (\ref{Cond3a}) into the second equation in (\ref{Cond2a}) yields
the following equation
$$
\delta_1^8 h^2 \rho^4 - \delta_1^4 (h^2 \rho^8 + 4) + 4 \rho^4 = 0
$$
with two roots $\delta_1^4 = \rho^4$ and $\delta_1^4 h^2 \rho^4 = 4$.
Since $\delta_1 = \rho$ implies $\cos \theta_1 = -1$ in (\ref{Cond3a}),
which is not admissible, we only have one positive root for $\delta_1$ given by
\begin{equation}
 \label{Cond3b}
\delta_1 = \frac{\sqrt{2}}{\rho \sqrt{h}},
\end{equation}
which implies
\begin{equation}
\label{Cond3c}
\cos \theta_1 = 2h \frac{1-\rho ^4}{4- h^2 \rho^4}
\end{equation}
thanks to (\ref{Cond3a}). Solutions (\ref{Cond3b}) and (\ref{Cond3c}) are equivalent to
(\ref{lam1}). The constraint (\ref{Cond2b}) with the solutions (\ref{Cond3b})--(\ref{Cond3c}) is rewritten in the form
$$
\frac{(1-\rho^4) (h^2 \rho^4 + 4)^2}{2 h \rho^2 (h^2 \rho^4 - 4)} > 0,
$$
from which the two regions in (\ref{lam2}) follow. In the exceptional case, $\delta_1 = 1$, we have from
the first equation in (\ref{Cond2a}) that $\rho = 1$ whereas $\cos \theta_1$ is not determined.
Then, the second equation in (\ref{Cond2a}) implies that $h = 2$ since $\cos \theta_1 = -1$ is not admissible.
The constraint (\ref{Cond2b}) yields $\cos \theta_1 < 0$ so that $\theta_1 \in \left(\frac{\pi}{2},\pi \right)$.

\section{Conclusion}

We have derived the one-fold Darboux transformation between solutions of the semi-discrete MTM system
using the Lax pair and the dressing methods. When one solution of the semi-discrete MTM system
is either zero or nonzero constant, the one-fold Darboux transformation generates one-soliton solution
on the zero or nonzero constant background respectively. When the one-fold Darboux transformation is
used recursively, it also allows us to construct two-soliton solutions and generally multi-soliton solutions.
We have showed that properties of the discrete solitons in the semi-discrete MTM system are very similar
to properties of the continuous MTM solitons.

Among further problems related to the semi-discrete MTM system, we mention construction of conserved quantities
which may clarify orbital stability of the discrete MTM solitons, similar to the work \cite{Yusuke2}. Another
direction is to develop the inverse scattering transform for solutions of the Cauchy problem
associated with the semi-discrete MTM system, similar to the work \cite{PelSaal}. Since 
numerical simulations of the semi-discrete MTM system (\ref{MTM-discrete}) present serious challenges, 
it may be interesting to look for another version of the integrable semi-discretization of the continuous MTM system (\ref{MTM}).

\vspace{0.25cm}

{\bf Acknowledgement.} The authors thank Leeor Greenblat for collaboration on numerical exploration
of the semi-discrete MTM system during an undergraduate research project. The work of TX was partially supported by the National Natural Science Foundation of China (No. 11705284) and the program of China Scholarship Council (No. 201806445009). TX also appreciates the hospitality of the Department of Mathematics \& Statistics at McMaster University during his visit in 2019. The work of DEP is supported by the State task program in the sphere
of scientific activity of Ministry of Education and Science of the Russian Federation
(Task No. 5.5176.2017/8.9) and from the grant of President of Russian Federation
for the leading scientific schools (NSH-2685.2018.5).

\end{document}